\documentclass{COMNET}

\usepackage{graphicx}%
\usepackage{amsmath,amssymb,amsfonts}%
\usepackage[table]{xcolor}%
\usepackage{url}
\usepackage{hyperref}
\usepackage{subfigure}
\usepackage{adjustbox}
\usepackage{colortbl}
\usepackage{array}
\usepackage{pifont}
\usepackage{cleveref}

\newtheorem{cor}[theorem]{\sc Corollary}
\newtheorem{lem}[theorem]{\sc Lemma}
\newtheorem{cl}[theorem]{\sc Claim}

\DeclareMathOperator*{\argmax}{\arg\!\max}

\begin{document}

\title{Percolation and localisation: Sub-leading eigenvalues of the nonbacktracking matrix}

\shorttitle{Percolation and localisation} 
\shortauthorlist{James Martin et al.} 

\author{
\name{James Martin$^\dagger$, Tim Rogers and Luca Zanetti}
\address{Department of Mathematical Sciences, University of Bath, Bath, UK\email{$^\dagger$Corresponding author. Email: jlm80@bath.ac.uk}}
}
\maketitle

\begin{abstract}
{
The spectrum of the nonbacktracking matrix associated to a network is known to contain fundamental information regarding percolation properties of the network. Indeed, the inverse of its leading eigenvalue is often used as an estimate for the percolation threshold.
However, for many networks with nonbacktracking centrality localised on a few nodes, such as networks with a core-periphery structure, this spectral approach badly underestimates the threshold.
In this work, we study networks that exhibit this localisation effect by looking beyond the leading eigenvalue and searching deeper into the spectrum of the nonbacktracking matrix.
We identify that, when localisation is present, 
the threshold often more closely aligns with the inverse of one of the sub-leading real eigenvalues: the largest real eigenvalue with a ``delocalised'' corresponding eigenvector.
We investigate a core-periphery network model and determine, both theoretically and experimentally, a regime of parameters for which our approach closely approximates the  threshold, while the estimate derived using the leading eigenvalue does not. We further present experimental results on large scale real-world networks that showcase the usefulness of our approach.
}
{percolation; localisation; nonbacktracking spectrum}
\end{abstract}

\section{Introduction}
Network percolation is a well-studied process within statistical physics that investigates the connectivity properties of random sub-networks. In bond (or site) percolation, edges (or nodes) are sampled independently at random with some occupation probability $p$, and the connectivity of the resulting subgraph is assessed. In many real-world networks, there is a sudden transition that separates two distinct regimes: a supercritical regime where the network is likely to contain a giant component that is said to percolate, and a subcritical regime where the network is more likely to be fragmented into many small non-percolating components. This transition has implications in many fields, such as the spread of infectious diseases, the robustness of communication or transport networks to sudden failures, or the permeability of porous materials \cite{hunt2014percolation, li2015network, li2015percolation, pastor2015epidemic}. Furthermore, the “critical” value $p_c$ that separates the two regimes is analogous to a phase boundary in a physical system, and estimating this percolation threshold is a key research question.
For infinite networks, the percolation threshold $p_c$ is a well-defined sharp transition that identifies the critical probability where a component of infinite size first emerges. As a percolating component is a component of infinite size, this concept is less meaningful on finite networks. However, when the network is large, a similar, but less distinct transition can often be observed and estimating this transition can provide critical insights into the connectivity of the network.

While some networks generated using well-defined models have analytical solutions to the percolation threshold, many real-world networks have complicated topological characteristics and analytical solutions are often absent \cite{mingli2021percolation, radicchi2015predicting}.
Instead, direct numerical simulations can precisely model the percolation process, with appropriate metrics to quantify it.
To simulate the percolation process, multiple Monte Carlo realisations are sampled and the percolation strength is computed as $P(p) = \langle S_1 \rangle/N$, where $\langle S_1 \rangle$ is the average number of nodes in the largest component when a proportion $p$ of the original number of edges are occupied~\cite{newman2001fast}.
There are several properties of the network that typically exhibit critical behaviour near the percolation threshold and can be computed by directly applying these numerical simulations. For example, the \textit{susceptibility} $\chi(p) = (\langle S_1^2 \rangle - \langle S_1 \rangle^2 )/ \langle S_1 \rangle$ of the network is expected to be maximal near the percolation threshold, as demonstrated in Figure~\ref{Fig:Background}, leading to the robust numerical estimate $p_\text{sus} = \argmax_p \; \chi(p)$ \cite{radicchi2015predicting}.
Alternatively, near the percolation threshold, we would expect the size of the non-percolating components to be largest, before they connect and form a giant percolating component. Therefore, another estimation of the percolation threshold is the occupation probability where the size of the second largest component is maximal, which we denote $p_\text{slc} = \argmax_p \langle S_2 \rangle$ \cite{karrer2014percolation}. Applying similar intuition, we would also expect the mean size of the non-percolating components to peak near the percolation threshold, leading to a further estimate $p_\text{mnp}$. 
For large networks, however, estimates that require direct simulations can be computationally expensive, and more efficient approaches are often desired.
\begin{figure}[t]
    \centering
    \subfigure[E-R network]{
    \includegraphics[width=0.42\textwidth]{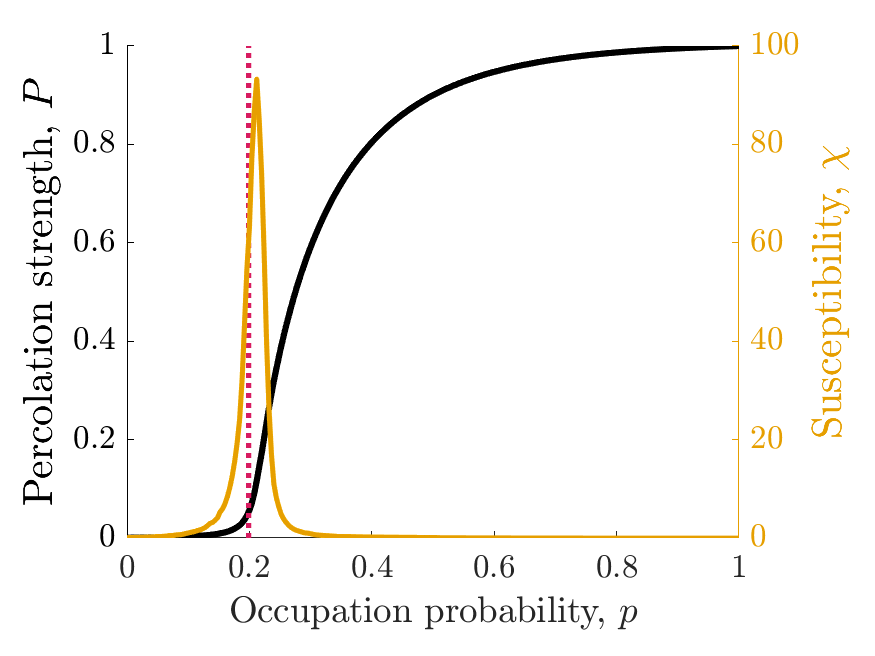}
    \label{fig:Background_ER}
    }
\hspace{0.3cm}
    \subfigure[Core-Periphery network]{
    \includegraphics[width=0.42\textwidth]{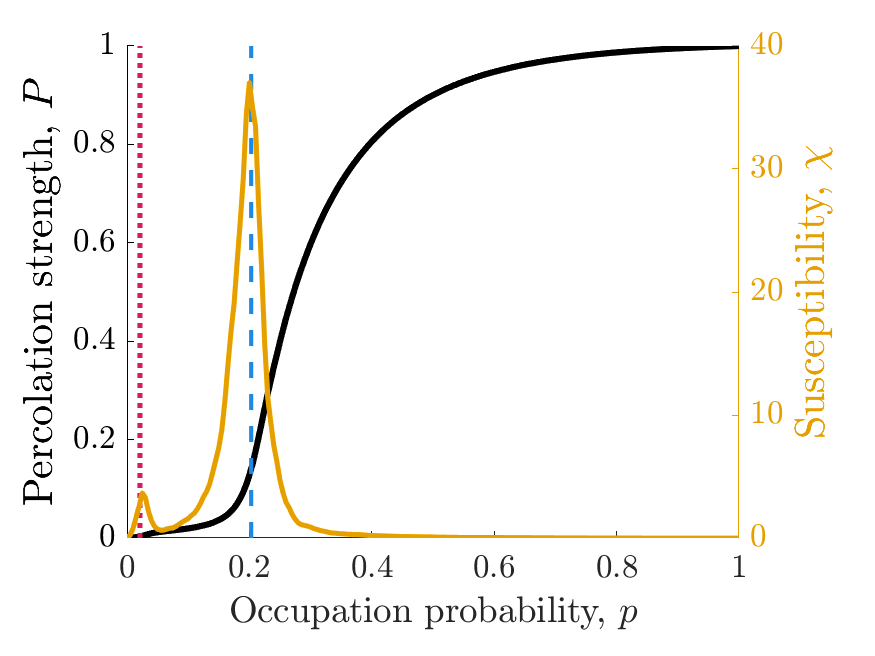}
    \label{fig:Background_CP}
    }
    \caption{
    An example of the percolation process for (a) an E-R network and (b) a core-periphery network. The transition in percolation strength $P(p)$ (thick solid black line) aligns with the maximum peak in susceptibility $\chi(p)$ (solid orange line). The estimate $p_\text{nbt} = 1/\lambda_1(B)$ (dotted vertical red line) aligns with the first peak, which is not maximal in a core-periphery network. The estimate $1/\lambda_2(B)$ (dashed vertical blue line) more closely predicts the threshold. 
    }
    \label{Fig:Background}
\end{figure}

In recent literature, message-passing techniques have shown that applying spectral methods to the so-called nonbacktracking matrix can uncover fundamental information about the network regarding properties such as   percolation or community structure \cite{krzakala2013spectral, hamilton2014tight, karrer2014percolation, budel2021detecting}. In this approach, the percolation threshold is predicted as $p_\text{nbt} = 1 / \lambda_1(B)$, where $B$ is the non-symmetric nonbacktracking matrix defined on the directed edges of the network such that $B_{i \to j,k \to l} = 1$ if $j=k$ and $i \neq l$, and zero otherwise~\cite{hashimoto1989zeta}. The leading eigenvalue $\lambda_1(B)$ measures the growth rate of the number of nonbacktracking walks, i.e., walks that do not immediately return to the previous node, and reflects the propagation of connections throughout the network. This concept closely relates to the branching number of an infinite tree, which describes the average number of new connections that each node generates, and determines whether the tree will continue to grow or eventually die out~\cite{timar2017nonbacktracking}.
In infinite locally tree-like networks, the leading eigenvalue of the nonbacktracking matrix captures the average branching factor: in particular, the network will likely contain a giant component if $\lambda_1(B)$ is greater than 1. By assigning a weight $p$ to each edge, reflecting the occupation probability, the percolation transition will occur close to where $\lambda_1(pB) =1$, and the estimate $p_\text{nbt}$ is recovered.
When the network contains cycles, however, some nonbacktracking walks can revisit edges, and $\lambda_1(B)$ overestimates the average branching factor. Then $p_\text{nbt}$ gives a lower bound of the actual percolation threshold, which could potentially be much smaller than the observed transition. 
Nevertheless, $p_\text{nbt}$ is an upper bound of the alternative spectral estimate $1/\lambda_1(A)$, where $A$ denotes the adjacency matrix. However, the latter method is often preferred when the network is dense, as both spectral estimates are close while the computation time is significantly reduced \cite{bollobas2010percolation}.

To reduce the computational complexity of the nonbacktracking method, the Ihara-Bass determinant formula can be exploited \cite{bass1992ihara}. Given an undirected network $G = (V,E)$, the spectrum of the $2|E| \times 2|E|$ nonbacktracking matrix $B$ contains the entire spectrum of the smaller $2|V| \times 2|V|$ \textit{reduced nonbacktracking matrix} $H$, defined as
\begin{equation}
\label{eq:H_def}
H =
\begin{pmatrix}
A & I-D \\ I & 0
\end{pmatrix},
\end{equation}
where $A$ denotes the symmetric adjacency matrix and $D$ denotes the degree matrix, with diagonal entries $D_{ii} = \sum_j A_{ij}$ and zeros otherwise.
The remaining eigenvalues of $B$ are exactly the trivial eigenvalues $-1$ and $1$, each with multiplicity $|E|-|V|$. Then
$$
p_\text{nbt} = \frac{1}{\lambda_1(B)} = \frac{1}{\lambda_1(H)}.
$$

This estimate closely predicts the percolation threshold for large sparse random networks which are locally tree-like. For example, consider an Erd{\H o}s-R{\' e}nyi (E-R) network $\mathcal{G}_{N,q}$, i.e., a network of $N$ nodes constructed by placing an edge between each pair of nodes independently with probability $q$. In the limit of large network size, the leading eigenvalue of the nonbacktracking matrix of an E-R network is the expected mean degree, i.e., $\lambda_1(B) = d \equiv (N-1)q$, while the remaining eigenvalues lie within a circle centred at the origin with radius $\sqrt{d}$ \cite{bordenave2015non}. Here, $p_\text{nbt}$ recovers the analytical solution to the percolation threshold $p_c = 1/d$ \cite{erdHos1960evolution}. Furthermore, this prediction closely aligns with the observed transition for finite networks, demonstrated in Figure~\ref{fig:Background_ER}.
 
In some networks, however, $p_\text{nbt}$ fails to predict the threshold. For example, when there exists small subgraphs within a large sparse network that contain anomalously high densities of edges, percolation estimates exhibit a spurious localisation effect \cite{pastor2020localization, timar2023localization}. 
The inverse of the leading eigenvalue of the nonbacktracking matrix is biased towards one of these subgraphs and the mass of the leading eigenvector is typically concentrated on the edges incident to it.
Edges within the dense subgraph do not contribute to connectivity propagation within the remainder of the network, which reduces the overall effective connectivity. The actual transition requires a much higher occupation probability to form a giant component spanning a large proportion of the network.

Furthermore, the estimate $p_\text{nbt}$ typically predicts the first peak of the susceptibility function, which tends not to be the maximal peak for networks when the nonbacktracking centrality (NBC) is concentrated on a small subset of nodes~\cite{pastor2020localization}. The NBC of a node $i$ is defined as
$
\text{NBC}_i = \sum_{j \in \partial_i} v_{1}^{(j \to i)}(B),
$
where $\partial_i = \{ j : A_{ij}=1 \}$ is the neighbourhood of $i$, and $v_{1}^{(j \to i)}(B)$ is the element in the leading eigenvector of $B$ that corresponds to the directed edge $j \to i$. This quantity is precisely the $i^\text{th}$ element of the leading eigenvector of $H$, that is
$
\text{NBC}_i = v_{1}^{(i)}(H),
$
and is proportional to the probability that $i$ is a member of the giant component near criticality~\cite{krzakala2013spectral}. Figure~\ref{fig:Background_CP} demonstrates this localisation effect on a network where a complete subgraph with $N=50$ and $d=49$ is embedded into a sparse E-R network with $N=5000$ and $d = 5$. The dynamics of the percolation process are very similar to the dynamics on the independent E-R network shown in Figure~\ref{fig:Background_ER}. However, $p_\text{nbt}$ badly underestimates the percolation transition and instead aligns with an additional initial peak in the susceptibility function. Conversely, the inverse of the next largest eigenvalue, $1/\lambda_2(B)$, more closely predicts the threshold.

In this paper, we further investigate the nonbacktracking spectrum and produce a methodology that more closely predicts the percolation threshold for a range of networks where $p_\text{nbt}$ typically fails. In particular, our methodology outperforms $p_\text{nbt}$ when small dense subgraphs are present in an otherwise sparse network. 
Real-world networks with this structure often have multiple real outliers in the nonbacktracking spectrum~\cite{krzakala2013spectral} and we hypothesize that exploring these additional sub-leading eigenpairs can lead to a more accurate prediction of the threshold. More precisely, our approach seeks the largest eigenvalue of the nonbacktracking matrix with a corresponding delocalised eigenvector, that is, an eigenvector whose mass is not concentrated on a small subset of elements. 
We present theoretical arguments for networks that are constructed by inserting a small number of edges between disconnected components. By bounding the perturbation of nonbacktracking eigenvalues upon inserting these edges, we identify a regime of parameters of a synthetic network model where our estimate will likely align with the percolation threshold, while outperforming $p_\text{nbt}$.

We structure this paper as follows.
In Section~\ref{Sec:Core-Periphery}, we summarise literature describing the effect that localisation has on predictions of the percolation threshold. We then formally introduce our approach to predict the percolation threshold, as well as present experimental and theoretical results on simple synthetically generated core-periphery networks.
In Section~\ref{Sec:Numerical Results}, we demonstrate the performance of our approach on more complex synthetic networks that contain multiple dense subgraphs, as well as showcasing results on some real-world networks reflecting collaboration on scientific research papers, as well as a peer-to-peer file sharing network.
Finally, in Section~\ref{Sec:Conclusion}, we discuss the message of the paper and some possible future research directions.

\section{Core-Periphery networks} \label{Sec:Core-Periphery} 
A simple example where the leading eigenvalue of the nonbacktracking matrix does not closely align with the percolation threshold is when the network has a core-periphery structure, in which there exists a small subgraph (the core) has a much higher density of edges than the remainder of the network (the periphery). 

\subsection{Related work}
When the NBC is concentrated on a small subset of nodes, Pastor‑Satorras and Castellano~\cite{pastor2020localization} observe that the first peak in susceptibility is often not maximal: instead, the first peak is located at the critical probability for the resulting induced subgraph to become connected. More generally, multiple peaks of the susceptibility often align with multiple percolation transitions within the network. In core-periphery networks, the NBC is typically concentrated on nodes within the core subgraph. The estimate $p_\text{nbt}$ aligns with the first peak in susceptibility, which predicts the critical probability for the core to become connected, badly underestimating the actual threshold of the network. 

To quantify this claim, let $\mathbf{u}_i(H)$ denote the first $N$ elements of the eigenvector $\mathbf{v}_i(H)$. Then entries in $\mathbf{u}_1(H)$ are proportional to the NBC of the corresponding node. 
From the definition of $H$, each eigenvector $\mathbf{v}_i$, with corresponding eigenvalue $\lambda_i$, has the form $\mathbf{v}_i = \begin{bmatrix} \mathbf{u}_i & \lambda^{-1}_i \mathbf{u}_i \end{bmatrix}^\intercal$, for some $\mathbf{u}_i \in \mathbb{C}^N$. Thus, when $\lambda_i$ is large, the last $N$ elements of the corresponding eigenvector will not contain any significant additional information regarding the percolation process. Therefore, in general, whenever the mass of the NBC is concentrated on a small subset of nodes, the vectors $\mathbf{u}_1(H),\mathbf{v}_1(H)$ and $\mathbf{v}_1(B)$ will be localised. 
One measure of localisation within a vector is the inverse participation ratio (IPR), defined as $\text{IPR}(\mathbf{x}) = (\sum_i |x_i|^4)/(\sum_i |x_i|^2)^2$~\cite{timar2023localization}, where $\text{IPR}(\mathbf{x})$ is lowest when all entries of $\mathbf{x}$ are equal and the vector is completely delocalised.

Consider a network model constructed by embedding an arbitrary finite (child) subgraph into an infinite locally tree-like (mother) network. Assume that branches emerging from the child network do not intersect at any finite distance. Let $B$, $B_\text{child}$ and $B_\text{mother}$ denote the nonbacktracking matrices of the composite network, the child component and the mother component, respectively. Timar et al.~\cite{timar2023localization} show that the leading eigenvalue of $B$ is the largest eigenvalue between $B_\text{child}$ and $B_\text{mother}$, i.e., $\lambda_1(B) = \max \left(\lambda_1(B_\text{child}),\lambda_1(B_\text{mother})\right)$. As the mother network is locally tree-like, the eigenvalue $\lambda_1(B_\text{mother})$ is simply its branching number. The percolation threshold of the composite network then depends on this branching number, which, in general, is a real positive eigenvalue within the spectrum of $B$, i.e., $p_c = 1/\lambda_1(B_\text{mother}) = 1/\lambda_i(B)$ for some $i$. When the finite child subgraph is sufficiently dense, using $\lambda_1(B)$ to predict the threshold will instead provide the underestimation $p_\text{nbt} = 1/\lambda_1(B_\text{child})$. Furthermore, the mass of the NBC of the network is concentrated on the subset of nodes within the dense child network. In particular, the vector $\mathbf{u}_1(H)$ is localised.
    
Krzakala et al.~\cite{krzakala2013spectral} investigate the spectrum of the nonbacktracking matrix of a sparse network. They show that the number of real eigenvalues outside of the circle with radius $\sqrt{\lambda_1(B)}$, which typically contains the bulk of the eigenvalues, often reflects the number of communities within the network. Furthermore, the corresponding eigenvectors can be used to partition the nodes into these communities. Real-world networks often have many real eigenvalues outside the bulk. With this in mind, a large sparse network that contains multiple small dense subgraphs can be viewed as having a community structure, and each subgraph contributes an additional large real eigenvalue outside of the bulk. When the subgraph is small, the corresponding eigenvector is localised on the edges incident to this induced subgraph. There may then be multiple large eigenvalues, with localised eigenvectors, whose inverse leads towards a prediction of the threshold biased towards the corresponding small subgraph. The actual threshold will more likely align with the inverse of the largest eigenvalue that corresponds to a large majority of nodes in the network, as reflected by the distribution of mass within the corresponding eigenvector. As the outlying eigenvalues are often real, these eigenpairs can be computed efficiently by using variations of the power method accompanied with deflation techniques.

\subsection{Our approach}
Suppose, in a core-periphery network, the periphery is locally tree-like and sufficiently large. Then the network resembles the composite network constructed by Timar et al.~\cite{timar2023localization}, where a finite network is embedded into an infinite locally tree-like network. This analogy suggests $p_\text{nbt}$ is biased towards connectivity properties of the core, and the NBC, i.e., $\mathbf{u}_1(H)$, is localised on the nodes within it.
The leading eigenvectors $\mathbf{v}_1(B)$ and $\mathbf{v}_1(H)$ will likely also be localised on the nodes or edges incident to the core. Furthermore, their result suggests the percolation threshold of the network will likely be closer to threshold of the subgraph induced by nodes in the periphery, which aligns with the inverse of the leading eigenvalue of its nonbacktracking matrix. This closely corresponds to the second largest real eigenvalue of the nonbacktracking matrix of the overall core-periphery network, that is $p_c \approx 1/\lambda_2(B)$. Moreover, the corresponding eigenvectors $\mathbf{u}_2(H)$, $\mathbf{v}_2(H)$ and $\mathbf{v}_2(B)$ will likely be ``delocalised'', i.e., their mass is distributed evenly across a large proportion of entries.

We extend this intuition to a large sparse network that contains multiple small dense subgraphs, which we observe typically have multiple real outliers within the nonbacktracking spectrum. 
Our approach seeks a delocalised eigenvector of the nonbacktracking matrix, so that the corresponding eigenvalue more fairly reflects the dynamics of a large majority of the network. 
Let $\lambda_1(B),\dots,\lambda_k(B)$ and $\mathbf{v}_1(B),\dots,\mathbf{v}_k(B)$ denote the positive real eigenvalues of the nonbacktracking matrix $B$ and their corresponding eigenvectors, respectively, where we sort eigenvalues in nonincreasing order (i.e., $\lambda_1(B) \geq \lambda_2(B) \geq \dots > 0$). We predict the percolation threshold will align with $p_\text{deloc} = 1/\lambda_j(H)$, where the corresponding vector $\mathbf{u}_j(H)$ is the first vector with a ``small'' IPR.
We do not prescribe a general precise rule to choose the appropriate eigenvalue, and instead use heuristics within our experiments.

In the remainder of Section~\ref{Sec:Core-Periphery}, we restrict our investigation to a simple core-periphery model, however we present experimental results for a synthetic network model with multiple cores, as well as on some real-world networks, in Section~\ref{Sec:Numerical Results}.

\subsection{Defining a synthetic core-periphery network model}
We investigate the performance of our approach on a synthetically generated network, constructed using the stochastic block model (SBM). The SBM is a widely used method to generate networks with an underlying community structure. The model generates a network $G$ by partitioning the nodes into disjoint groups and assigning edges between pairs of nodes depending on this assignment. The network is constructed by assigning each node $i$ a label $l_i \in \{1,\dots,R\}$, such that there are $N_r$ nodes labeled $r$. The size of the network is then $N = \sum_{r=1}^R N_r$. Let $Q$ be an $R \times R$ symmetric matrix of probabilities. Then, an edge is independently generated between every unordered pair of nodes $i$ and $j$ with probability $Q_{l_i,l_j}$.

Consider a network $G$ constructed using the SBM with block sizes $N_1 = N_\text{d}$, $N_2 = N_\text{s}$ and probability matrix 
\begin{equation}
\label{Eq:SBM_CP}
Q = 
\begin{pmatrix}
    q_\text{d} & q_\text{s} \\ q_\text{s} & q_\text{s}
\end{pmatrix},
\end{equation}
where $N_\text{d} \ll N_\text{s}$. 
The subgraphs $G_\text{d}$ and $G_\text{s}$ induced by the diagonal blocks are Erd{\H o}s-R{\' e}nyi (E-R) networks with mean degrees $d_\text{d} = (N_\text{d}-1)q_\text{d}$ and $d_\text{s} = (N_\text{s}-1)q_\text{s}$, respectively. When $d_\text{d} > d_\text{s}$, this construction models a core-periphery network. For a network constructed in this way, and with $N_\text{d}$ fixed while $N_\text{s} \to \infty$, the network will closely resemble an E-R network with mean degree $d_\text{s}$, which has the analytical solution $p_c = 1/d_\text{s}$~\cite{erdHos1960evolution}.

\subsection{Experimental results on a core-periphery network model}
In Figure~\ref{Fig:Syn_SBM_1Dense}, we present experimental results for a network constructed using the SBM with a core-periphery structure defined using \eqref{Eq:SBM_CP}. The two leading eigenvalues of $B$ closely correspond to the leading eigenvalues of the nonbacktracking matrix of $G_\text{d}$ and $G_\text{s}$, i.e., $\lambda_1(H) \approx d_\text{d}$ and $\lambda_2(H) \approx d_\text{s}$. Our approach $p_\text{deloc} = 1/\lambda_2(H)$ aligns with $p_\text{sus}, p_\text{slc}$ and $p_\text{mnp}$, which are each derived using numerical simulations, and are all located where the percolation strength appears to undergo a transition. The estimate $p_\text{nbt} = 1/\lambda_1(H)$ clearly underestimates the threshold.
Additionally, the leading vector $\mathbf{u}_1(H)$ is much more localised than $\mathbf{u}_2(H)$, as reflected by a decrease in IPR when sorting the real eigenvalues in decreasing order (i.e., from right to left). 
\begin{figure}[t]
    \centering
    \includegraphics[width=0.42\textwidth]{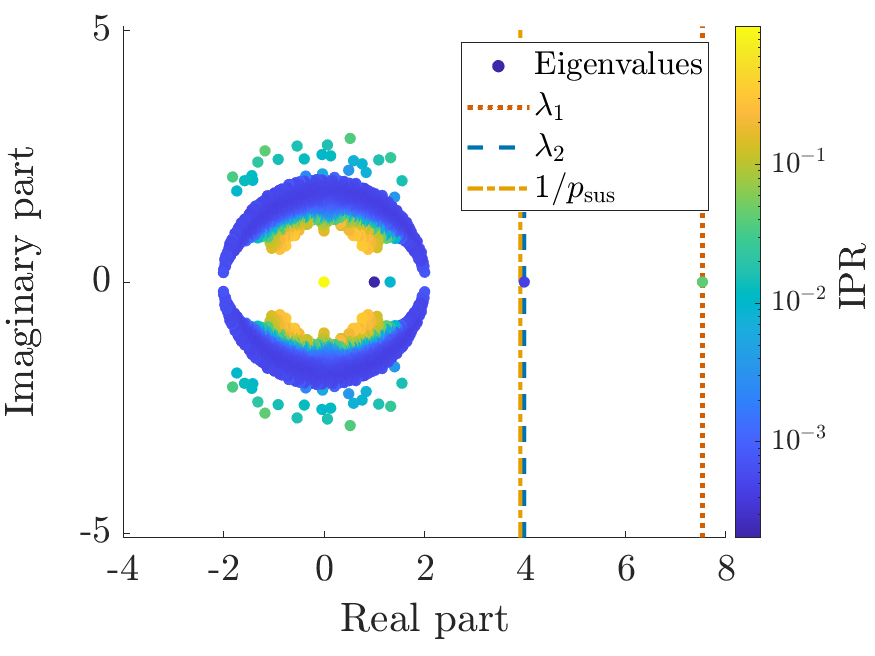}
    \hspace{0.3cm}
    \includegraphics[width=0.42\textwidth]{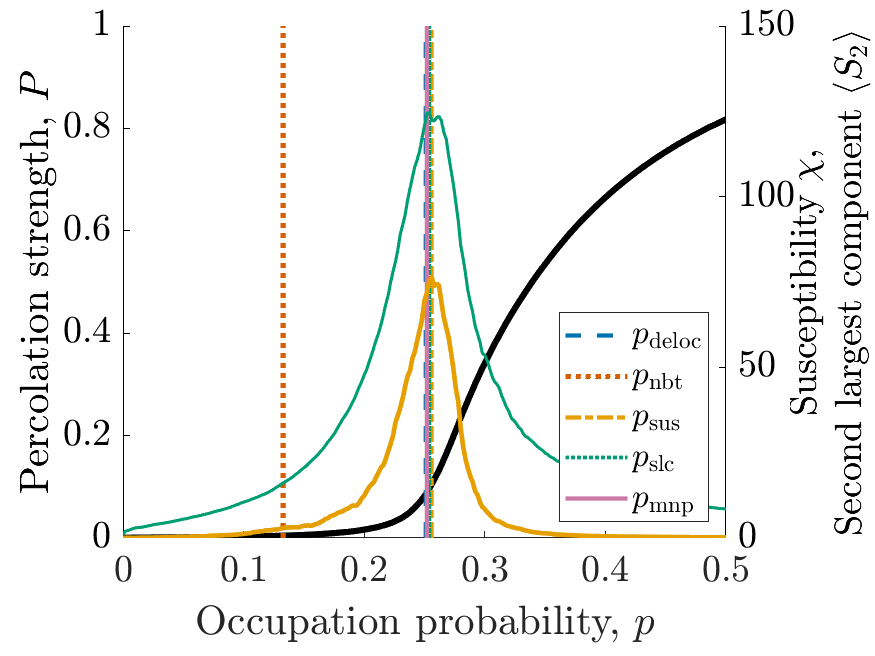}

    \includegraphics[width=0.42\textwidth]{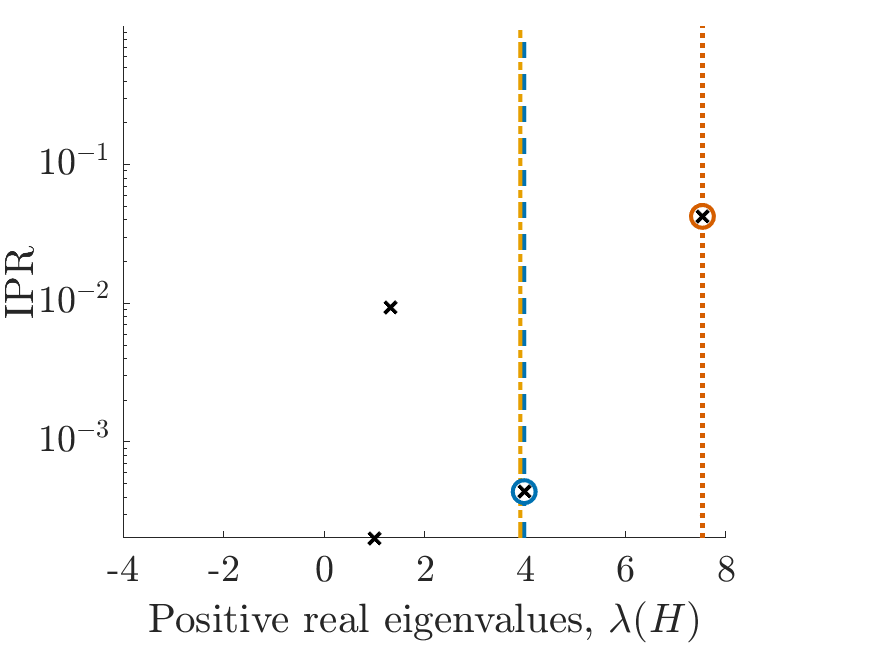}
    \hspace{0.3cm}
    \includegraphics[width=0.42\textwidth]{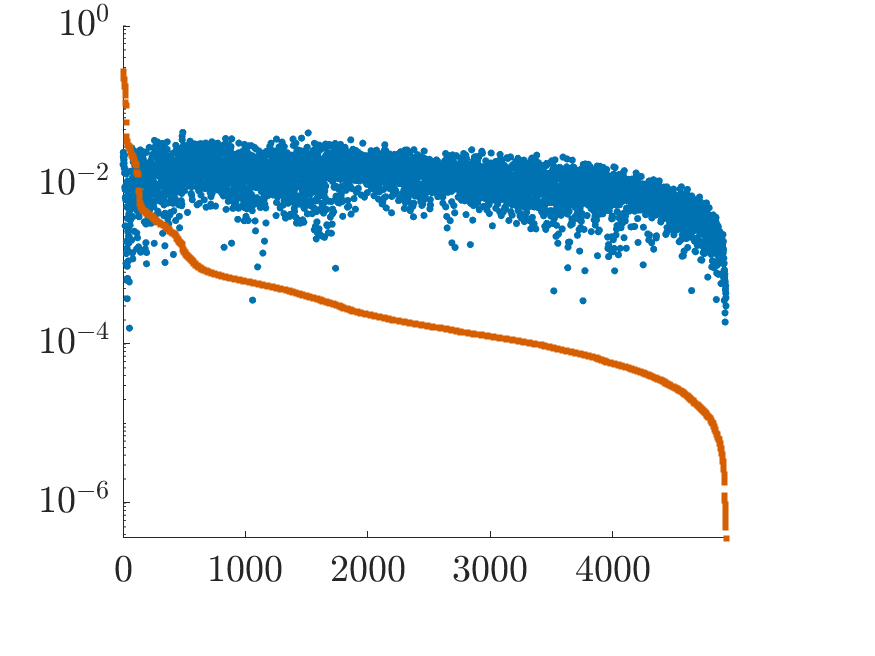}
    \caption{
    Analysis of an SBM network, with $N_\text{d}=25$, $N_\text{s} = 5000$, $d_\text{d} = 8$ and $d_\text{s} = 4$. 
    Top left: Eigenvalues of $H$, coloured according to the IPR of the corresponding vector $\mathbf{u}(H)$.
    Top right: Percolation strength $P$ (thick black line), susceptibility $\chi$ (solid orange line) and largest non-percolating component size $\langle S_2 \rangle$ (thin solid green line), over 500 Monte Carlo simulations. Numerical estimates are recorded in Table~\ref{Tab:Results}.
    Bottom left: IPR of the vectors $\mathbf{u}(H)$ corresponding to real positive eigenvalues. Circles indicate vectors $\mathbf{u}_{1}(H)$ (red) and $\mathbf{u}_{2}(H)$ (blue).
    Bottom right: Absolute value of entries of $\mathbf{u}_1(H)$ (red squares) and $\mathbf{u}_2(H)$ (blue circles), sorted in descending order of entries in $\mathbf{u}_1(H)$.
    }
    \label{Fig:Syn_SBM_1Dense}
\end{figure}

Recall that the bulk eigenvalues in the nonbacktracking spectrum of the induced core component $G_\text{d}$ are expected to be confined to a circle with radius $\sqrt{d_\text{d}}$.
Figure~\ref{Fig:Syn_SBM_1Dense} shows that the bulk eigenvalues corresponding to the core subgraph within the composite network $G$ also appear to be confined to this circle. 
If this core subgraph is sufficiently dense, such that $\sqrt{d_\text{d}} > d_\text{s}$, some of these bulk eigenvalues will likely have a larger magnitude than the eigenvalue corresponding to $d_\text{s}$, potentially causing this desired eigenvalue to become more challenging to detect.\footnote{See Appendix~\ref{App:Extra_phase} for additional experimental results on the IPR when sorting eigenvalues in descending order of magnitude.}
However, the figure shows that the bulk core eigenvalues are typically complex and their corresponding eigenvectors are localised: choosing the largest delocalised real eigenvalue will likely recover the largest eigenvalue corresponding to the periphery component. 
\begin{figure}[t]
    \centering
    \includegraphics[width=0.45\textwidth]{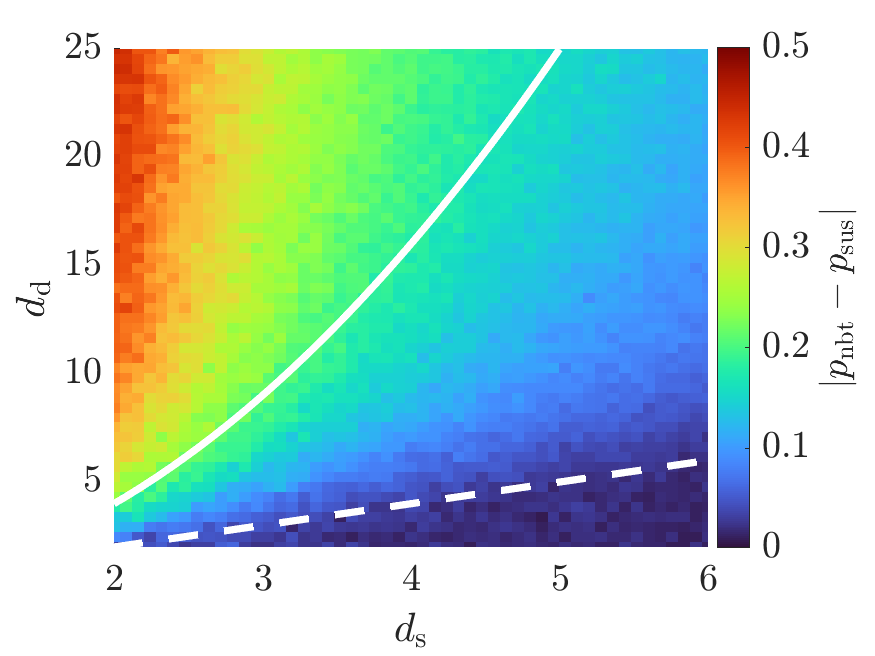}
\hspace{0.1cm}
    \includegraphics[width=0.45\textwidth]{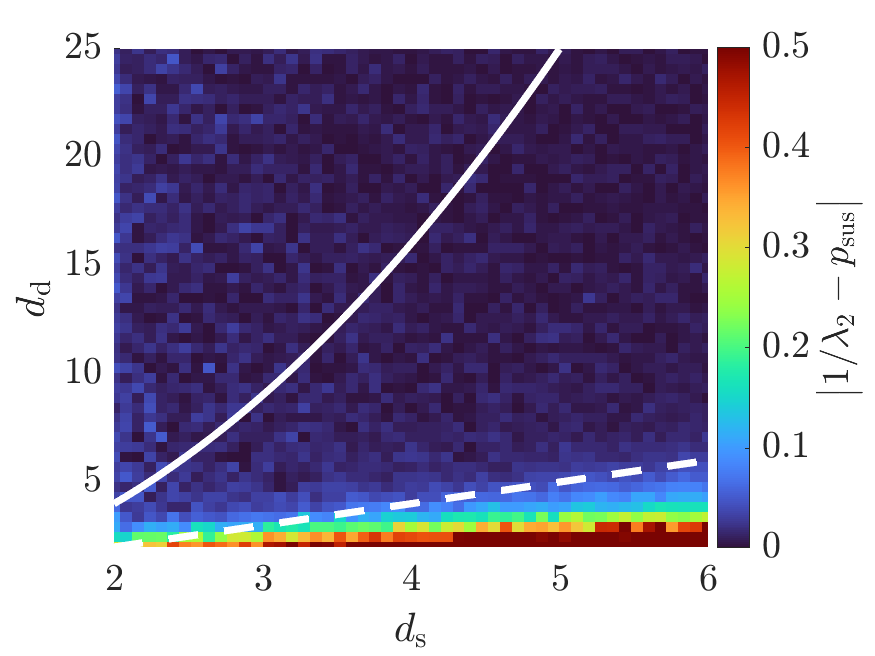}
    \caption{
    Results on an SBM network, with $N_\text{d}=25$ and $N_\text{s}=1500$ and a range of parameters for $d_\text{d}$ and $d_\text{s}$. For each set of parameters, the estimates $p_\text{nbt}$ and $1/\lambda_2$ are averaged over 50 randomly generated networks. The estimate $p_\text{sus}$ is averaged over $5$ randomly generated networks, using 20 Monte Carlo simulations per network.
    Difference between the inverse of the largest (left) or second largest (right) real eigenvalue and the susceptibility estimate, i.e., $\big| p_\text{nbt} - p_\mathrm{sus} \big|$ or $\big| 1/\lambda_2  - p_\mathrm{sus} \big|$.
    Solid line: $d_\text{s} = \sqrt{d_\text{d}}$. Dashed line: $d_\text{s} = d_\text{d}$.
    }
    \label{Fig:Syn_Phase}
\end{figure}

Figure~\ref{Fig:Syn_Phase} displays a measure for the error in spectral estimations by comparing the result with numerical simulations. The figure shows that, for networks in the core-periphery regime $d_\text{s} > d_\text{d}$, the estimate $p_\text{deloc} = 1 / \lambda_2(H)$ closely aligns with $p_\text{sus}$ and largely outperforms $p_\text{nbt}$. In this regime, experimental results of the IPR suggest $\mathbf{u}_2(H)$ is much less localised than $\mathbf{u}_1(H)$: our method will indeed appropriately select the eigenvalue $\lambda_2(H)$ to use within the estimation of the threshold. We expect our approach to perform particularly well when $\sqrt{d_\text{d}} < d_\text{s}$, however the figure further shows that outside of this regime (i.e., above the solid line), $p_\text{deloc}$ can also often precisely predict the percolation threshold.

While Figure~\ref{Fig:Syn_Phase} shows a strong performance in the regime $d_\text{s} > \sqrt{d_\text{d}}$, we note that certain parameters might result in a pair of real eigenvalues becoming complex conjugate pairs, potentially causing the eigenvalue relevant to the percolation threshold to become undetectable. An investigation into this complex dynamic of the nonbacktracking spectrum is left for further research.

\subsection{A bound for the perturbation of eigenvalues}
We now identify a regime of networks such that the nonbacktracking eigenvalue with second largest magnitude closely corresponds to the leading nonbacktracking eigenvalue of the underlying periphery component.
We apply the perturbation arguments of Coste and Zhu~\cite{coste2021eigenvalues} to obtain a theorem that bounds the perturbation of the leading nonbacktracking eigenvalues of the disconnected core and periphery components when small number of edges are inserted between them.

\begin{theorem}
\label{Thrm: Eigenvalue perturbation}
    Let $A, \hat{A}, X, \hat{X}$ be $N \times N$ matrices. Let $\hat{A}$ and $\hat{X}$ be block diagonal matrices with blocks $\hat{A}_1, \dots, \hat{A}_l$ and $\hat{X}_1,\dots,\hat{X}_l$ respectively, where each $\hat{A}_i$ and $\hat{X}_i$ are $N_i \times N_i$ matrices such that $N = \sum_i N_i$. Suppose each $\hat{A}_i$ and $\hat{X}_i$ are commonly diagonalisable by the non-singular matrix $P_i$, i.e., $P_i \hat{A}_i P_i^{-1}$ and $P_i \hat{X}_i P_i^{-1}$ are diagonal matrices. Let
    \begin{equation*}
        L =
        \begin{pmatrix}
            A & X \\ I & 0
        \end{pmatrix}
        \hspace{0.8cm} \text{and} \hspace{0.8cm} 
        \hat{L} =
        \begin{pmatrix}
            \hat{A} & \hat{X} \\ I & 0
        \end{pmatrix}.
    \end{equation*}
    Then, for any eigenvalue $\mu$ of $L$, there exists an eigenvalue $\nu$ of $\hat{L}$ such that
    \begin{equation*}
        | \mu - \nu | \leq \sqrt{\max_i \kappa (P_i)} \cdot \sqrt{\| X-\hat{X} + \mu (A-\hat{A})\|_2},
    \end{equation*}
    where $\kappa(P_i)$ is the condition number $\kappa(P_i) = \| P_i \|_2 \cdot \| P_i^{-1} \|_2$.
\end{theorem}

This theorem is an extension of Theorem~2.2 from Coste and Zhu~\cite{coste2021eigenvalues}, and becomes equivalent with $l=1$. A proof of Theorem~\ref{Thrm: Eigenvalue perturbation} is provided in Appendix~\ref{App:Theorem_proof}.
A useful corollary is as follows.
\begin{cor}
\label{cor:eval_perturbation}
    Suppose $\hat{X}$ is diagonal such that each $\hat{X}_i$ has the form $c_i I$ for some scalar $c_i$. Suppose $\hat{A}$ is Hermitian. Then, for any eigenvalue $\mu$ of $L$, there exists an eigenvalue $\nu$ of $L_0$ such that
    \begin{equation*}
        |\mu - \nu| \leq \sqrt{\| X - \hat{X} + \mu ( A-\hat{A}) \|_2}.
    \end{equation*}
\end{cor}
\begin{proof}
As $\hat{A}$ is Hermitian, each block $\hat{A}_i$ is Hermitian, hence diagonalisable by a unitary matrix $P_i$. As each $\hat{X}_i$ is diagonal, each block $\hat{A}_i$ and $\hat{X}_i$ are commonly diagonalisable $P_i$. As $P_i$ is unitary, then $\kappa(P_i) = \|P_i \|_2 \cdot \|P_i^{-1} \|_2 = 1$ for all $i$. By applying Theorem~\ref{Thrm: Eigenvalue perturbation}, we get the desired result.
\end{proof} \vspace{0.2cm}

We will apply \cref{cor:eval_perturbation} to analyse the following random model for core-periphery networks.
Let $G_\text{d} = \mathcal{G}_{N_\text{d},q_\text{d}}$ be a ``small and dense'' E-R network and $G_\text{s} = \mathcal{G}_{N_\text{s},q_\text{s}}$ be a ``large and sparse'' E-R network. We denote with $A_\text{d}$, $A_\text{s}$ their adjacency matrices  and with $d_\text{d} = (N_\text{d} - 1) q_\text{d}$, $d_\text{s} = (N_\text{s} - 1) q_\text{s}$ their expected degree. We assume $N_\text{d} \ll N_\text{s}$ and  $d_\text{d} > d_\text{s}$, so that $G_\text{d}$ will represent the \emph{core} of the network, while $G_\text{s}$ will represent the \emph{periphery}. 
%
We connect
the core to the periphery with a number of edges so that any node in $G_\text{d}$ is connected to at most $K$ nodes in $G_\text{s}$, while any node in $G_\text{s}$ is connected to at most $1$ node in $G_\text{d}$.
With this network model, we will show that the inverse of the second largest eigenvalue of the nonbacktracking matrix of $G$ is closer to the expected percolation threshold of the network $1/d_\text{s}$, rather than the inverse of its leading nonbacktracking eigenvalue $p_\text{nbt}$, which is closer to $1/d_\text{d}$.

Let $C \in \{0,1\}^{N_\text{d} \times N_\text{s}}$, where the entry $C_{ij}$ indicates if there exists an edge between the $i$-th vertex of the core and the $j$-th vertex of the periphery. The adjacency matrix $A$ of $G$ can then be represented as
\begin{equation*}
A =
\begin{pmatrix}
A_\text{d} & C \\ C^{\text{T}} & A_\text{s}
\end{pmatrix}.
\end{equation*}

We recall the reduced nonbacktracking matrix $H$, which contains the non-trivial nonbacktracking eigenvalues of $G$ and, for the purpose of this analysis, we define a matrix $\hat{H}$, such that
\begin{equation*}
    H =
    \begin{pmatrix}
        A_\text{d} & C & I-D_\text{d} & 0 \\
        C^T & A_\text{s} & 0 & I-D_\text{s} \\
        I & 0 & 0 & 0 \\
        0 & I & 0 & 0
    \end{pmatrix}
    \hspace{0.5cm} \text{and} \hspace{0.5cm} 
    \hat{H} =
    \begin{pmatrix}
        A_\text{d} & 0 & (1-d_\text{d})I & 0 \\
        0 & A_\text{s} & 0 & (1-d_\text{s})I \\
        I & 0 & 0 & 0 \\
        0 & I & 0 & 0
    \end{pmatrix}.
\end{equation*}
Notice that the matrix $\hat{H}$ excludes edges between the the components $G_\text{d}$ and $G_\text{s}$ defined by $C$, and uses expected degree instead of the observed degree. 
With this construction, the spectrum of $\hat{H}$ can be determined exactly, given the spectrum of $A_\text{d}$ and $A_\text{s}$~\cite{coste2021eigenvalues}. 

Assume $G_\text{d}$ and $G_\text{s}$ are at least semi-sparse, that is, $q_\text{d}>\log(N_\text{d})/(N_\text{d}-1)$ and $q_\text{s}>\log(N_\text{s})/(N_\text{s}-1)$. Additionally, let $K \leq 4$. By applying Corollary~\ref{cor:eval_perturbation}, we will show that, if $9 \sqrt{d_\text{d}} < d_\text{s} < d_\text{d}$, the eigenvalue $\lambda_2(H)$ is close to the eigenvalue $\lambda_1(H_\text{s}) \approx d_\text{s}$, while $\lambda_1(H)$ is closer to the eigenvalue $\lambda_1(H_\text{d}) \approx d_\text{d}$.

We start by observing that Corollary~\ref{cor:eval_perturbation} implies that, for any eigenvalue $\mu$ of $H$, there exists an eigenvalue $\nu$ of $\hat{H}$ such that
\begin{align}
\begin{split}
\label{eq:applying_bound}
    |\mu - \nu| &\leq
    \sqrt{
    \begin{Vmatrix}
        D_\text{d} - d_\text{d} I & 0 \\ 0 & D_\text{s} - d_\text{s} I
    \end{Vmatrix}
     _2 + |\mu| 
     \begin{Vmatrix}
        0 & C \\ C^T & 0
    \end{Vmatrix}
    _2 }.
\end{split}
\end{align}
By the definition of the spectral norm,
\begin{align}
\label{eq:boundC}
    \begin{Vmatrix} 0 & C \\ C^T & 0 \end{Vmatrix}
    _2
    =
    \sqrt{ 
    \begin{Vmatrix} CC^T & 0 \\ 0 & C^TC \end{Vmatrix}
    _2}
    =
    \sqrt{
    \max \left( \| CC^T\|_2, \|C^TC\|_2 \right)
    }
    \leq
    \sqrt{K} \leq 2,
\end{align}
which follows because $C$ has no more than $K$ elements in either row and no more than $1$ element in either column.

The first term under the square root of \eqref{eq:applying_bound} is bounded by $\max(\|D_\text{d}-d_\text{d} I \|_2, \|D_\text{s} - d_\text{s} I \|_2)$, which we evaluate using a Chernoff bound.

\begin{lem}[Chernoff Bound. Corollary~4.6 in \cite{mitzenmacher2017probability}]
\label{lem:chernoff}
    Let $Y_1, \dots, Y_n$ be independent random variables with $Y_i \in \{0,1\}$. Let $Y = \sum_i Y_i$ and $\mathbb{E}[Y]$ be the expected value of $Y$. Then, for any $0 < \delta < 1$, we have 
$$
\mathbb{P} \left( \left| Y - \mathbb{E}[Y] \right| \geq \delta \mathbb{E}[Y] \right)
\leq 2 \exp \left( \frac{ - \delta^2 \mathbb{E}[Y]}{3} \right).
$$
\end{lem}

We use Lemma~\ref{lem:chernoff} to first bound $\|D_\text{d}-d_\text{d} I \|_2$. Recall that $G_\text{d}$ has expected degree $d_\text{d} = (N_\text{d} - 1) q_\text{d}$. Denote the maximum degree of this network as $\Delta(G_\text{d})$ and let the network be semi-sparse, so that $q_\text{d} > \log N_\text{d} / (N_\text{d}-1)$. We apply the Chernoff bound using $\delta \geq 2\sqrt{\log N_\text{d} / (N_\text{d}-1)q_\text{d}}$, and then apply the union bound to give
\begin{equation*}
    \mathbb{P} \left( \left| \Delta(G_\text{d}) - d_\text{d} \right| \geq 2\sqrt{d_\text{d} \log N_\text{d}} \right)
    \leq
    2 N_\text{d} \exp \left( \frac{ - 4\log N_\text{d} }{3} \right)
    =
    2N_\text{d}^{-\frac{1}{3}}.
\end{equation*}
As $N_\text{d}$ increases, the RHS becomes small. Then, for large enough $N_\text{d}$, we can say with high probability (w.h.p.) that $\|D_\text{d}-d_\text{d} I \|_2  = | \Delta(G_\text{d}) - d_\text{d} |< 2\sqrt{d_\text{d} \log N_\text{d}} \leq 2 d_\text{d}$.  Following similar arguments, we find $\|D_\text{s}-d_\text{s} I \|_2 < 2 d_\text{s}$
Then, w.h.p., $\max(\|D_\text{d}-d_\text{d} I \|_2, \|D_\text{s} - d_\text{s} I \|_2) < 2d_\text{d}$.

By combining this result with \eqref{eq:boundC}, equation~\eqref{eq:applying_bound} becomes
 $$
 |\mu - \nu| \leq \sqrt{2d_\text{d} + 2|\mu|}.
 $$
To simplify our analysis, we impose the bound $|\mu| \leq 7d_\text{d}$, which follows by considering the row sums of $H$ under the assumption $ \Delta(G_\text{s}) < \Delta(G_\text{d}) \leq (7 d_\text{d} - 2K + 1)/2$. Then we have that
\begin{align*}
|\mu - \nu| &<
4 \sqrt{d_\text{d}}.
\end{align*}

To use Corollary~\ref{cor:eval_perturbation}, we need to identify $\nu$. By applying Lemma~2.1 from Coste and Zhu~\cite{coste2021eigenvalues}, we find that the eigenvalues of $\hat{H}$ can be computed exactly from the eigenvalues of the adjacency matrix $A$. For our network model, $\lambda_2(\hat{H}) = \left(1 + o(1)\right) \lambda_1(H_\text{s}) \approx d_\text{s}$ and $\lambda_3(\hat{H}) = \left(1 + o(1)\right) \lambda_2(H_\text{d}) \approx \sqrt{d_\text{d}}$.
Denote the eigenvalues $\mu_{\text{s},1} = \lambda_2(H)$, $\mu_{\text{d},2} = \lambda_3(H)$ and $\nu_{\text{s},1} = \lambda_2(\hat{H})$, $\nu_{\text{d},2} = \lambda_3(\hat{H})$.
We want to identify the regime of networks where $\mu_{\text{s},1} > \mu_{\text{d},2}$ given that $\nu_{\text{s},1} > \nu_{\text{d},2}$, i.e., we seek the perturbation of $\nu_{\text{s},1}$ and $\nu_{\text{d},2}$ to be small, after the edges indicated by $C$ are inserted between $G_\text{d}$ and $G_\text{s}$.
We have that, w.h.p.,
\begin{align*}
    |\mu_{\text{s},1} - d_\text{s}|
    <
    4\sqrt{d_\text{d}}
\hspace{0.5cm} \text{
and
} \hspace{0.5cm}
    |\mu_{\text{d},2} - \sqrt{d_\text{d}}|
    <
    4\sqrt{d_\text{d}}.
\end{align*}
The distance between $\mu_{\text{s},1}$ and $\mu_{\text{d},2}$ then satisfies $d_\text{s} - \sqrt{d_\text{d}} > 8\sqrt{d_\text{d}}$. Therefore, in the regime $9 \sqrt{d_\text{d}} < d_\text{s}$, which further requires $d_\text{s} > 81$, the eigenvalue $\lambda_2(H)$ aligns closely with the eigenvalue $\lambda_2(\hat{H})$, which recall aligns closely with $\lambda_1(H_\text{s})$. The estimate $1/\lambda_2(H)$ will likely outperform $p_\text{nbt}$ for networks that satisfy these conditions.

We remark that when performing a range of experiments, the observed eigenvalue perturbation appears significantly smaller than the bound suggests. While we do not provide a tighter bound for the eigenvalue perturbation, we do not exclude the possibility of one, and we leave this for potential future research.

\section{Numerical results} \label{Sec:Numerical Results}

To showcase the potential of our observation, we extend our approach to more complex networks by providing experimental results on synthetic networks containing multiple dense subgraphs, as well as on a range of real-world networks where the existence of small dense subgraphs is less clear. Numerical results for these experiments are collated in Table~\ref{Tab:Results}.

\newcommand*{\vcorr}{
    \vadjust{\vspace{-\dp\csname @arstrutbox\endcsname}\vspace{-0.01cm}}%
    \global\let\vcorr\relax
}
\newcommand*{\HeadAux}[3]{%
      \multicolumn{1}{@{}r@{}}{%
      \vcorr
        \vspace{-0.001cm}
        \sbox2{\hspace{4.8cm}}
        \sbox4{0}
        \rlap{
        \hspace{-2\tabcolsep}
        \hspace{-#2}
          \raisebox{\wd4}{\rotatebox{#3}{#1}}
        }
        \ifx\HeadLine Y
          \dimen0=\dimexpr\wd2
          \rlap{\hspace{-0.1cm} \rotatebox{#3}{\hbox{\vrule width\dimen0 height .4pt}}}
        \fi
      }
    }
    \newcommand*{\head}[3]{\HeadAux{\global\let\HeadLine=Y#1}{#2}{#3}}
    \newcommand*{\headNoLine}[3]{\HeadAux{\global\let\HeadLine=N#1}{#2}{#3}}
    \newcolumntype{x}[1]{>{\centering\arraybackslash}p{#1}}
    \newcommand\rotation{60}

\begin{table}[!h]
    \caption{
    Numerical estimates of the percolation threshold for a range of synthetic and real-world networks, corresponding to more detailed experimental results shown in Figures~\ref{Fig:Syn_SBM_1Dense},~\ref{Fig:Syn_SBM_5Dense}-\ref{Fig:RW_Gnutella}. The three estimates $p_\text{sus}$, $p_\text{slc}$ and $p_\text{mnp}$ are computed using direct simulations. The spectral estimates are $p_\text{nbt}=1/\lambda_1(B)$ and $p_\text{deloc} = 1/\lambda_j(B)$, where $B$ is the nonbacktracking matrix. More details on the derivation of these estimates can be found within the text. In the GRQC network, the small vertical line separates two estimates produced using different choices of sub-leading eigenvalue.
    }
    \vspace{0.4cm}
    \setlength{\tabcolsep}{6pt}
    \centering
    \begin{tabular}
    {x{3cm}|x{0.6cm}|x{0.6cm}|x{0.6cm}|x{0.6cm}|x{1.5cm}|x{1cm}}
        \head{Network}{2.1cm}{0}&
        \head{Susceptibility, $p_\text{sus}$}{0.3cm}{\rotation} &
        \head{Second largest component, $p_\text{slc}$}{0.3cm}{\rotation} &
        \head{\shortstack{Mean non-percolating \\ components, $p_\text{mnp}$}}{0.4cm}{\rotation} &
        \head{\shortstack{Nonbacktracking spectral \\ estimate, $p_\text{nbt}$}}{0.4cm}{\rotation} &
        \head{\shortstack{Delocalised spectral \\ estimate, $p_\text{deloc}$}}{0.85cm}{\rotation} &
        \headNoLine{Eigenvalue index, $j$}{0.5cm}{\rotation}
        \\
        \hline
        SBM, $r=1$ & 0.26 & 0.25 & 0.25 & 0.13 & 0.25 & 2
        \\
        SBM, $r=5$ & 0.25 & 0.25 & 0.25 & 0.15 & 0.25 & 6
        \\
        HEP-TH & 0.08 & 0.08 & 0.18 & 0.03 & 0.06 & 4
        \\
        GRQC & 0.13 & 0.13 & 0.16 & 0.02 & 0.10 $|$ 0.12 & 14 $|$ 19
        \\
        Gnutella & 0.05 & 0.11 & 0.11 & 0.04 & 0.12 & 4
        \\ \hline
    \end{tabular}
    \label{Tab:Results}
\end{table}

\subsection{Synthetic networks}
We construct a synthetic network containing $r$ dense embedded subgraphs, generated using the SBM with block sizes $N_1 = \dots = N_{r}=N_\text{d}$, $N_{r+1}=N_\text{s} \gg N_\text{d}$ and the symmetric $r+1 \times r+1$ probability matrix $Q_r$ where
\begin{equation*}
Q_r(i,j) = 
\begin{cases}
    q_\text{d} & \text{if $i = j \leq r$}
    \\
    q_\text{s} & \text{otherwise}
\end{cases}.
\end{equation*}
The case $r=1$ reflects the core-periphery model used in Section~\ref{Sec:Core-Periphery} and the trivial case $r=0$ is an Erd{\H o}s-R{\' e}nyi model $\mathcal{G}_{N_\text{s},q_\text{s}}$.
\begin{figure}[t]
    \centering
    \includegraphics[width=0.42\textwidth]{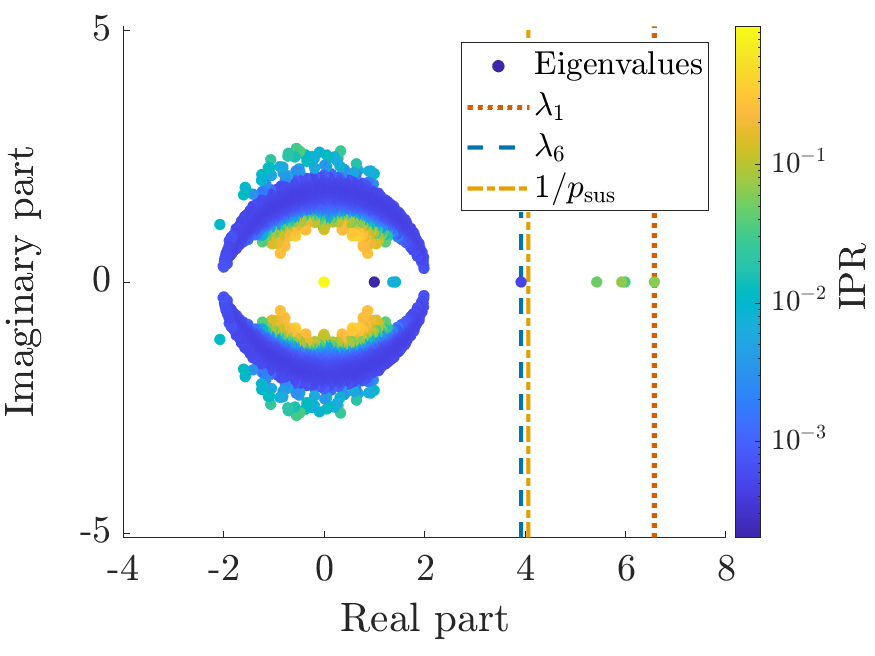}
    \hspace{0.3cm}
    \includegraphics[width=0.42\textwidth]{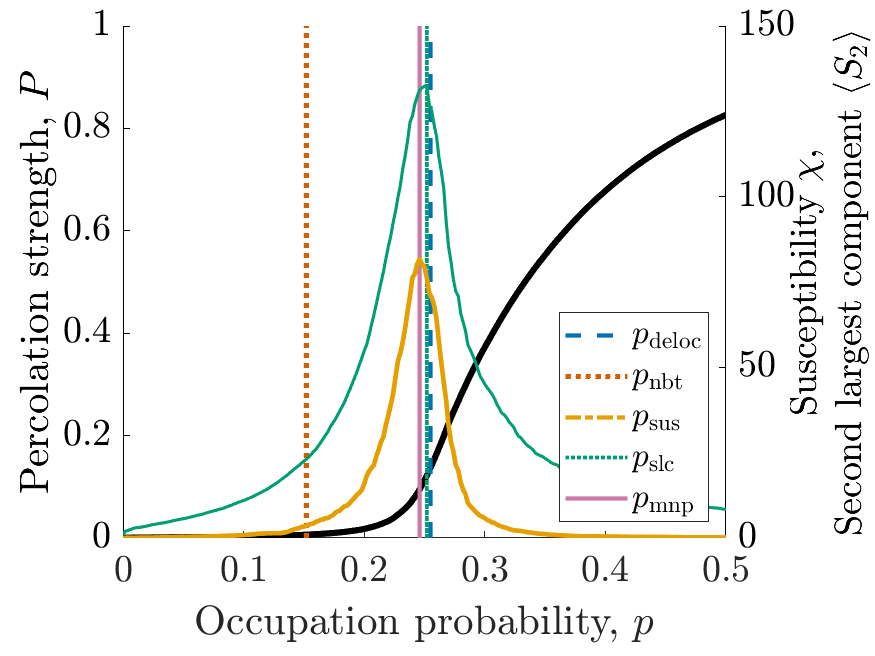}

    \includegraphics[width=0.42\textwidth]{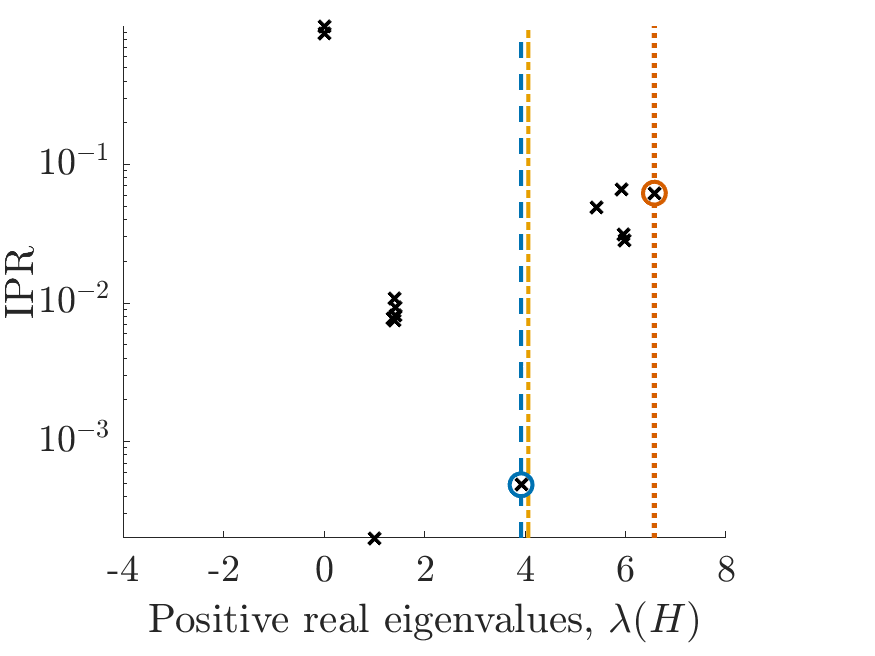}
    \hspace{0.3cm}
    \includegraphics[width=0.42\textwidth]{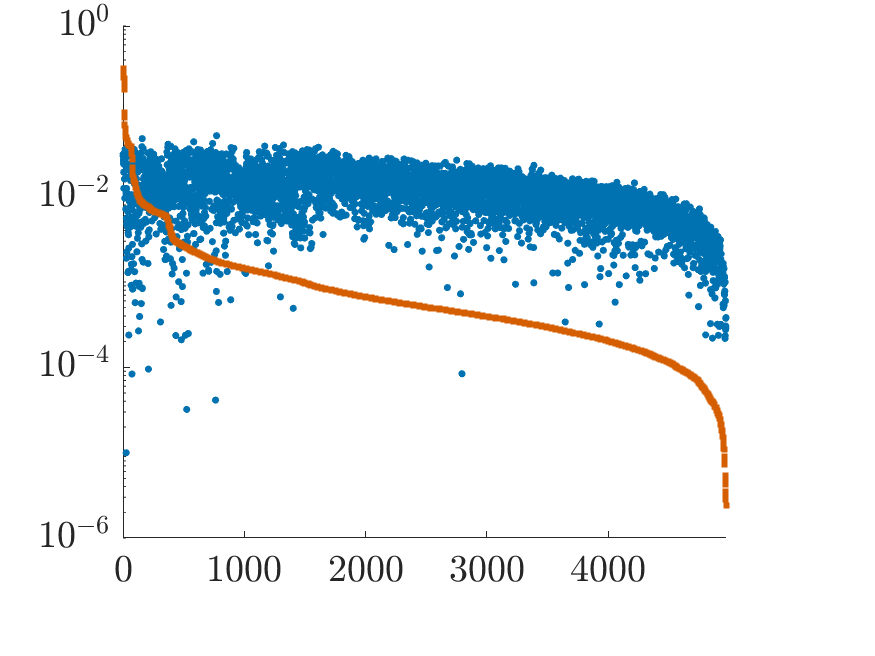}
    \caption{
    Analysis of an SBM network using probability matrix $Q_5$, with $N_\text{d}=12$, $N_\text{s} = 5000$, $d_\text{d} = 8$ and $d_\text{s} = 4$.
    Top left: Eigenvalues of $H$, coloured according to the IPR of the corresponding vector $\mathbf{u}(H)$.
    Top right: Percolation strength $P$ (thick black line), susceptibility $\chi$ (solid orange line) and largest non-percolating component size $\langle S_2 \rangle$ (thin solid green line), over 500 Monte Carlo simulations. Numerical estimates are recorded in Table~\ref{Tab:Results}.
    Bottom left: IPR of the vectors $\mathbf{u}(H)$ corresponding to real positive eigenvalues. Circles indicate vectors $\mathbf{u}_{1}(H)$ (red) and $\mathbf{u}_{6}(H)$ (blue).
    Bottom right: Absolute value of entries of $\mathbf{u}_1(H)$ (red squares) and $\mathbf{u}_{6}(H)$ (blue circles), sorted in descending order of entries in $\mathbf{u}_1(H)$.
    }
    \label{Fig:Syn_SBM_5Dense}
\end{figure}

Figure~\ref{Fig:Syn_SBM_5Dense} shows a synthetically generated network using $Q_5$, where 5 dense subgraphs are embedded into a much larger sparse network. The sixth eigenvalue best corresponds to the percolation threshold and the corresponding eigenvector is considerably less localised than the first five eigenvectors, as shown by the large decrease in IPR as well as a visual observation of the eigenvectors.

\subsection{Real-world networks}
We provide some experimental results on real-world networks to showcase our approach. For these complex networks, there often exists an eigenvalue whose inverse closely aligns with $p_\text{sus}$, where there is a sudden sharp increase in percolation strength. While our method of exploring additional sub-leading eigenvalues does not outperform $p_\text{nbt}$ for every network, we particularly observe an improvement when the nonbacktracking centrality is localised.
\begin{figure}[t]
    \centering
    \includegraphics[width=0.42\textwidth]{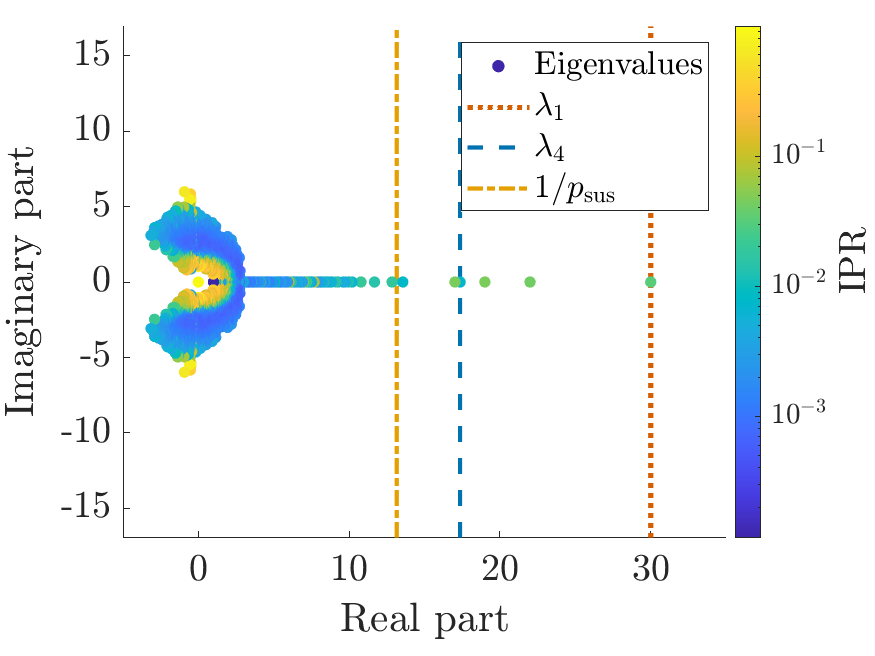}
    \hspace{0.3cm}
    \includegraphics[width=0.42\textwidth]{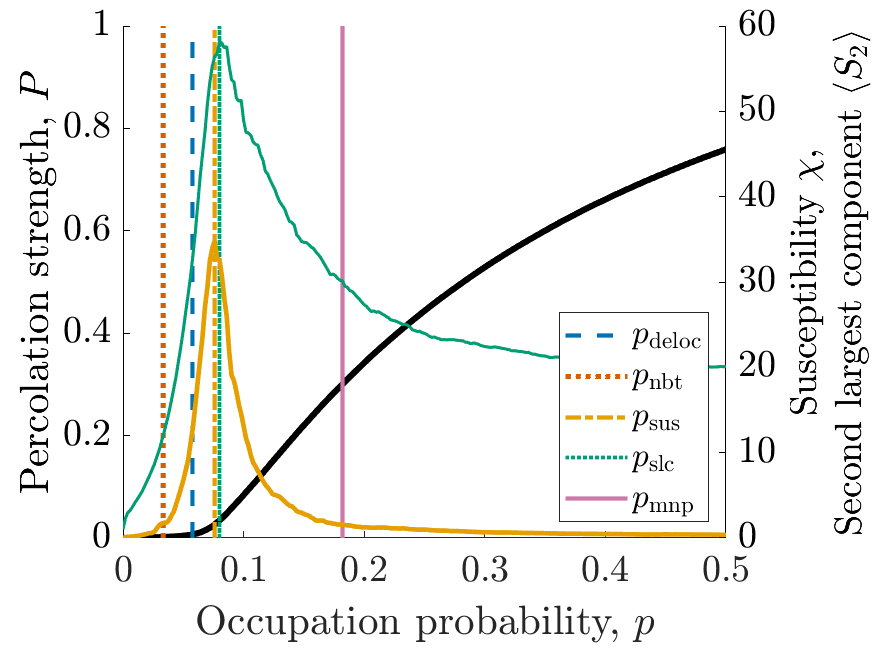}

    \includegraphics[width=0.42\textwidth]{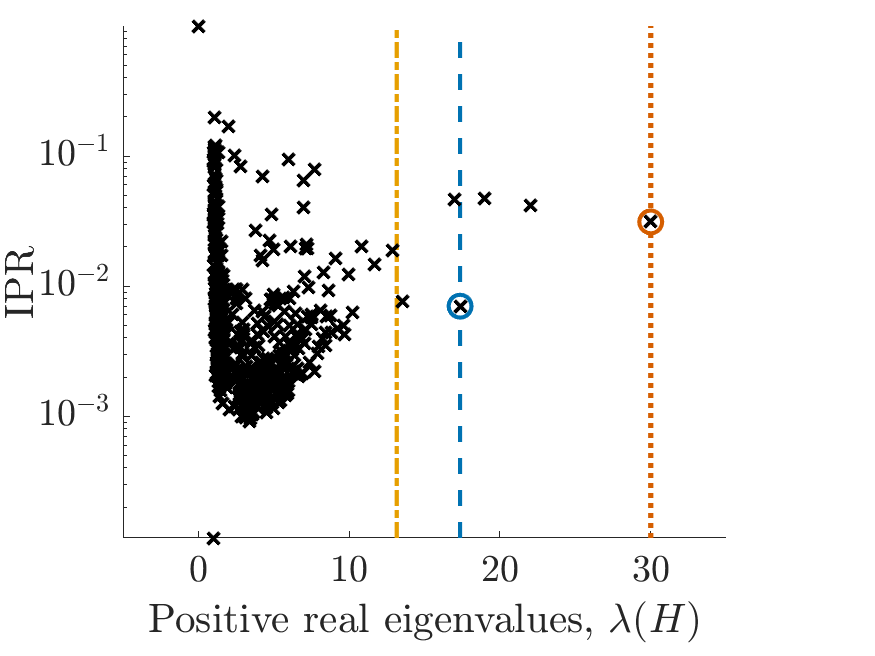}
    \hspace{0.3cm}
    \includegraphics[width=0.42\textwidth]{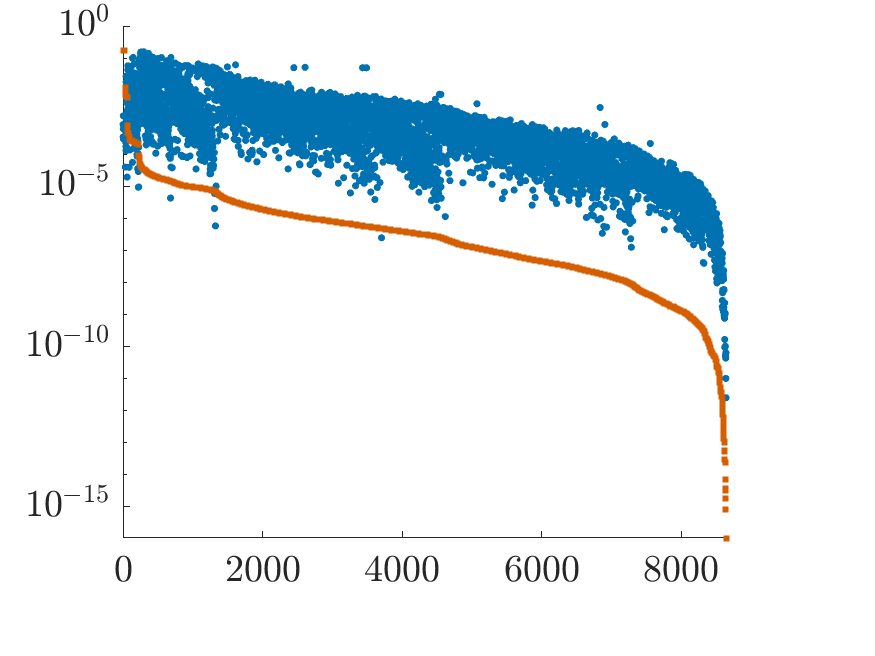}
    \caption{
    Analysis of a collaboration network for papers submitted to the High Energy Physics Theory section of the arXiv.
    Top left: Eigenvalues of $H$, coloured according to the IPR of the corresponding vector $\mathbf{u}(H)$.
    Top right: Percolation strength $P$ (thick black line), susceptibility $\chi$ (solid orange line) and largest non-percolating component size $\langle S_2 \rangle$ (thin solid green line), over 500 Monte Carlo simulations. Numerical estimates are recorded in Table~\ref{Tab:Results}.
    Bottom left: IPR of the vectors $\mathbf{u}(H)$ corresponding to real positive eigenvalues. Circles indicate vectors $\mathbf{u}_{1}(H)$ (red) and $\mathbf{u}_{4}(H)$ (blue).
    Bottom right: Absolute value of entries of $\mathbf{u}_1(H)$ (red squares) and $\mathbf{u}_4(H)$ (blue circles), sorted in descending order of entries in $\mathbf{u}_1(H)$.
    }
    \label{Fig:RW_HepTh}
\end{figure}

We first present results on a real-world dataset modelling collaborations between authors that submitted to the High Energy Physics Theory (HEP-TH) section of the arXiv~\cite{leskovec2007graph,snapnets}. An undirected edge between two nodes implies the corresponding authors co-authored at least one paper. After taking the largest connected component, the network contains $N=8638$ nodes and $M=24806$ edges. One example where localisation can arise in this type of network is when a paper has an unusually large authorship, thus forming a large completely connected subgraph. 
Figure~\ref{Fig:RW_HepTh} suggests there exists a percolation transition and that the nonbacktracking estimate $p_\text{nbt}$ underestimates this threshold. Indeed, the eigenvectors corresponding to the three largest real eigenvalues exhibit localisation, whereas the fourth does not: the inverse of fourth largest real eigenvalue aligns more closely with both the peak of the susceptibility function and the observed transition in percolation strength.
Furthermore, $\text{IPR}(\lambda_6)$ has a similar size to $\text{IPR}(\lambda_4)$, and the estimate $1/\lambda_6$ more strongly aligns with $p_\text{sus}$.
\begin{figure}[t]
    \centering
    \includegraphics[width=0.42\textwidth]{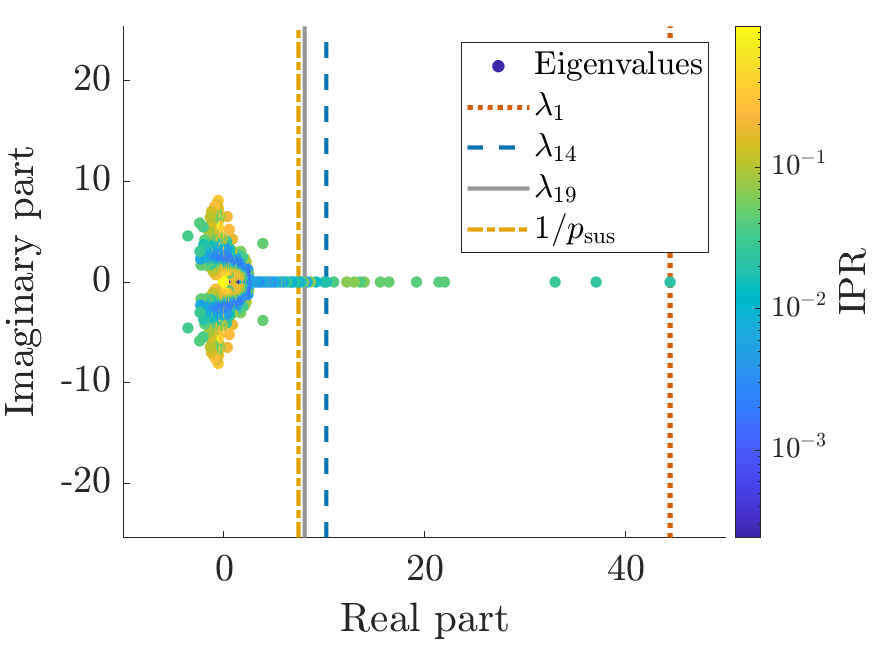}
    \hspace{0.3cm}
    \includegraphics[width=0.42\textwidth]{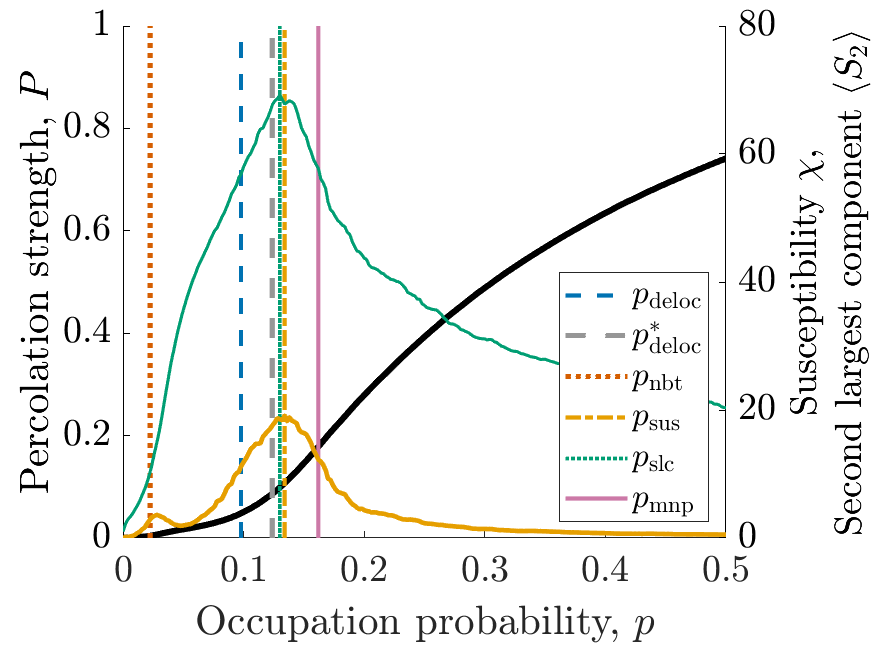}

    \includegraphics[width=0.42\textwidth]{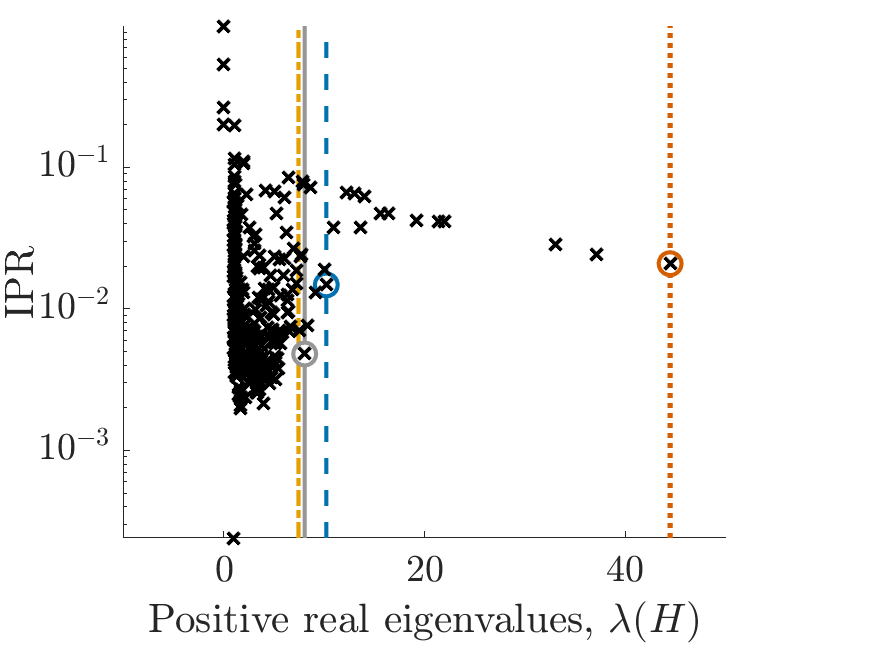}
    \hspace{0.3cm}
    \includegraphics[width=0.42\textwidth]{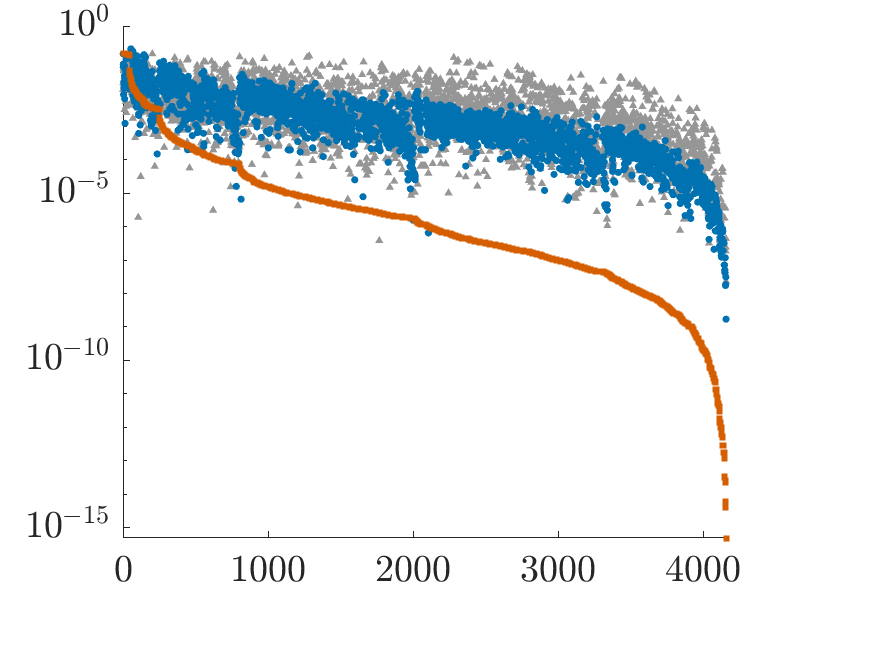}
    \caption{
    Analysis of a collaboration network for papers submitted to General Relativity and Quantum Cosmology.
    Top left: Eigenvalues of $H$, coloured according to the IPR of the corresponding vector $\mathbf{u}(H)$.
    Top right: Percolation strength $P$ (thick black line), susceptibility $\chi$ (solid orange line) and largest non-percolating component size $\langle S_2 \rangle$ (thin solid green line), over 500 Monte Carlo simulations. Numerical estimates are recorded in Table~\ref{Tab:Results}.
    Bottom left: IPR of the vectors $\mathbf{u}(H)$ corresponding to real positive eigenvalues. Circles indicate vectors $\mathbf{u}_{1}(H)$ (red), $\mathbf{u}_{14}(H)$ (blue) and $\mathbf{u}_{19}(H)$ (grey).
    Bottom right: Absolute value of entries of $\mathbf{u}_1(H)$ (red squares), $\mathbf{u}_{14}(H)$ (blue circles) and $\mathbf{u}_{19}(H)$ (grey triangles), sorted in descending order of entries in $\mathbf{u}_1(H)$.
    }
    \label{Fig:RW_GRQC}
\end{figure}

We next present experimental results on another real-world collaboration network, representing author collaborations on research papers submitted to General Relativity and Quantum Cosmology (GRQC)~\cite{snapnets}. After taking the largest component, the network contains $N=4158$ nodes and $M=13422$ edges.
Figure~\ref{Fig:RW_GRQC} demonstrates the nonbacktracking estimate $p_\text{nbt}$ is significantly smaller than the approaches that use direct simulations, and predicts the first peak of the susceptibility function rather than the maximal peak.
The IPR of $\mathbf{u}_{14}(H)$ is noticeably less than the previous vectors and the estimate $p_\text{deloc} = 1/\lambda_{14}$ is much closer to the maximal susceptibility peak, providing a better estimate. Furthermore, the following vectors exhibit additional decreases in IPR. In particular, we might alternatively argue the vector $\mathbf{u}_{19}(H)$ is the first vector to be delocalised. Indeed, the value $p^*_\text{deloc} = 1/\lambda_{19}$ is a better estimate, as it more closely aligns with the peak of the susceptibility function $p_\text{sus}$, as well as the estimate $p_\text{slc}$.
\begin{figure}[t]
    \centering
    \includegraphics[width=0.42\textwidth]{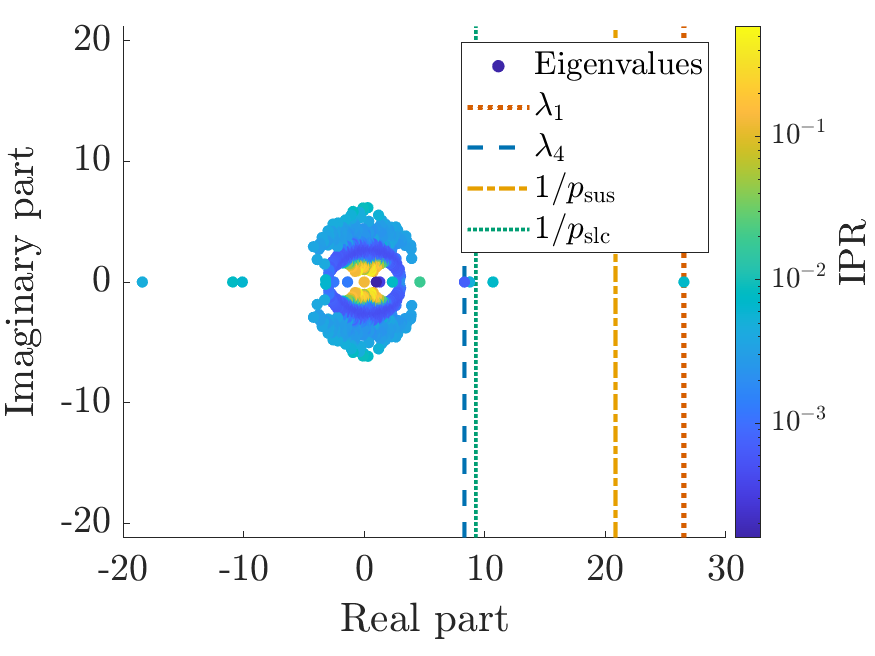}
    \hspace{0.3cm}
    \includegraphics[width=0.42\textwidth]{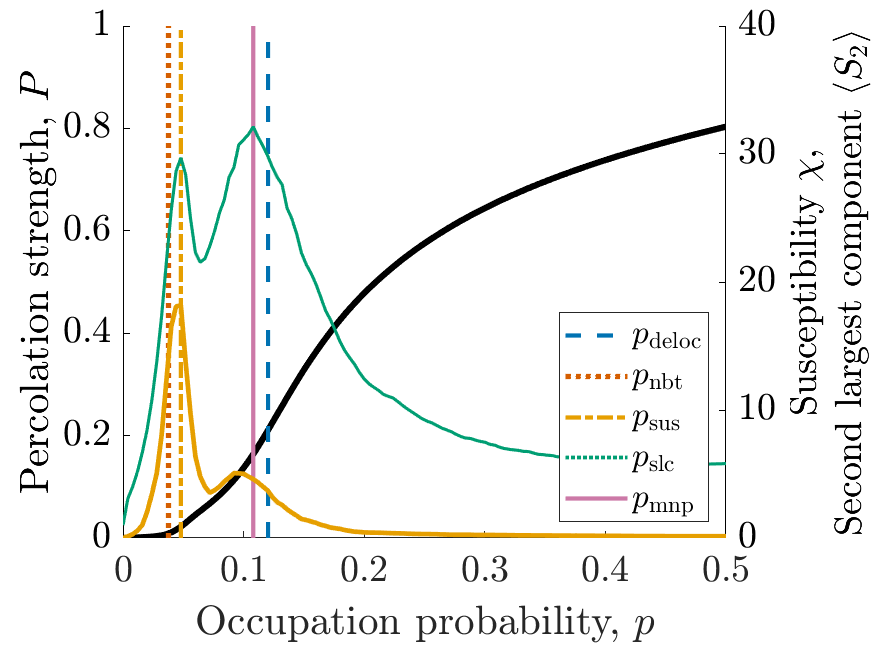}

    \includegraphics[width=0.42\textwidth]{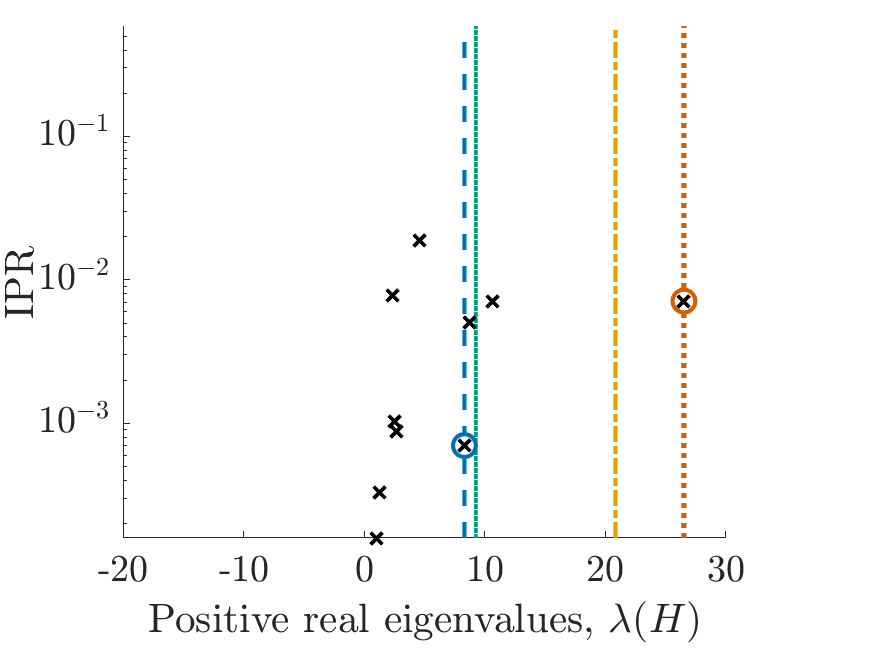}
    \hspace{0.3cm}
    \includegraphics[width=0.42\textwidth]{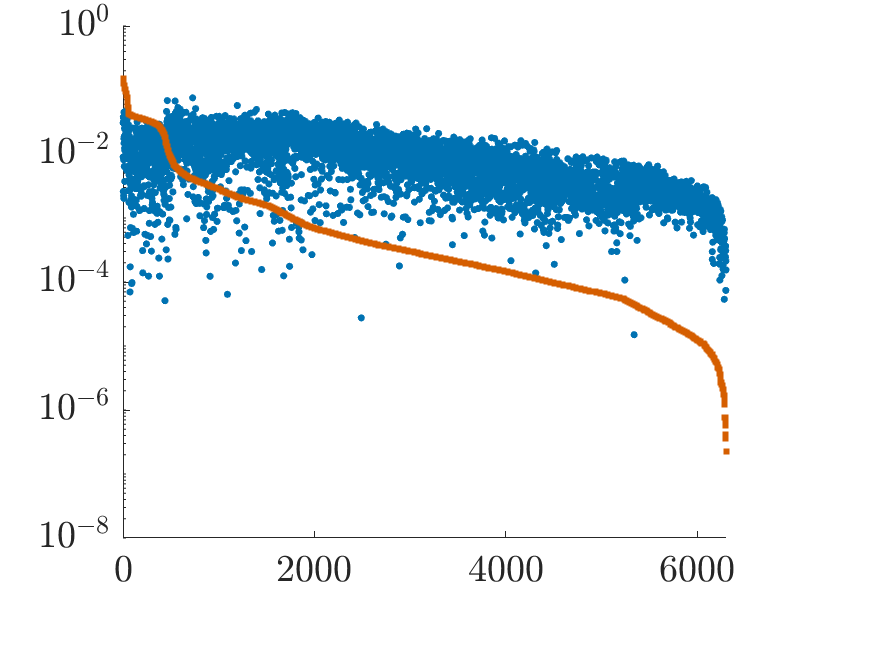}
    \caption{
    Analysis of a Gnutella peer-to-peer filesharing network.
    Top left: Eigenvalues of $H$, coloured according to the IPR of the corresponding vector $\mathbf{u}(H)$.
    Top right: Percolation strength $P$ (thick black line), susceptibility $\chi$ (solid orange line) and largest non-percolating component size $\langle S_2 \rangle$ (thin solid green line), over 500 Monte Carlo simulations. Numerical estimates are recorded in Table~\ref{Tab:Results}.
    Bottom left: IPR of the vectors $\mathbf{u}(H)$ corresponding to real positive eigenvalues. Circles indicate vectors $\mathbf{u}_{1}(H)$ (red) and $\mathbf{u}_{4}(H)$ (blue).
    Bottom right: Absolute value of entries of $\mathbf{u}_1(H)$ (red squares) and $\mathbf{u}_4(H)$ (blue triangles), sorted in descending order of entries in $\mathbf{u}_1(H)$.
    }
    \label{Fig:RW_Gnutella}
    \vspace{-0.3cm}
\end{figure}

Finally, we perform our algorithm on a Gnutella peer-to-peer filesharing network~\cite{ripeanu2002mapping,snapnets}. We construct an undirected network, where the $N$ nodes represent Gnutella hosts, while the $M$ undirected edges indicate a connection between them. We perform experiments on the largest connected component, where $N=6299$ and $M=20776$.
As shown in Figure~\ref{Fig:RW_Gnutella}, the IPR of the leading eigenvector is very low, suggesting that $p_\text{nbt}$ is a good prediction of the percolation threshold $p_c$, which is also supported experimentally. However, the eigenvector corresponding to the seventh largest real eigenvalue has a significantly smaller IPR than previous eigenvectors, and the reciprocal of this eigenvalue, $p_\text{deloc} = 1/\lambda_4$, aligns closely with the second peak of the susceptibility function, as well as the alternative predictions $p_\text{slc}$ and $p_\text{mnp}$. Furthermore, near this value, the slope of the percolation strength $P$ becomes (slightly) steeper. Some networks undergo multiple percolation transitions, leading to the susceptibility function exhibiting multiple peaks: we suggest that our approach might help us identify these additional transitions~\cite{pastor2020localization}.
\begin{figure}[t]
    \centering
    \includegraphics[width=0.32\textwidth,trim={1.6cm 0 2.1cm 0},clip]{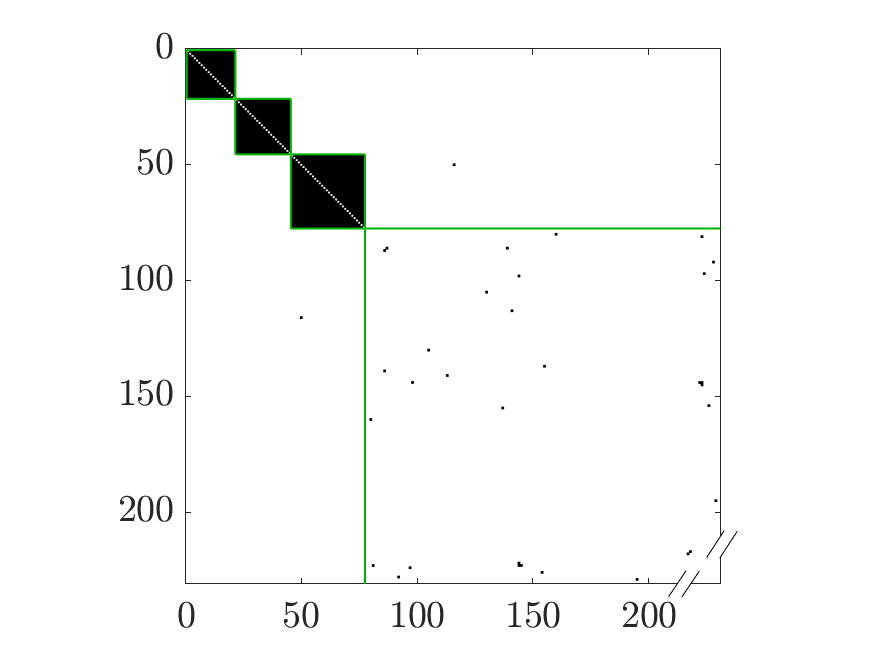}
    \hfill
    \includegraphics[width=0.32\textwidth,trim={1.6cm 0 2.1cm 0},clip]{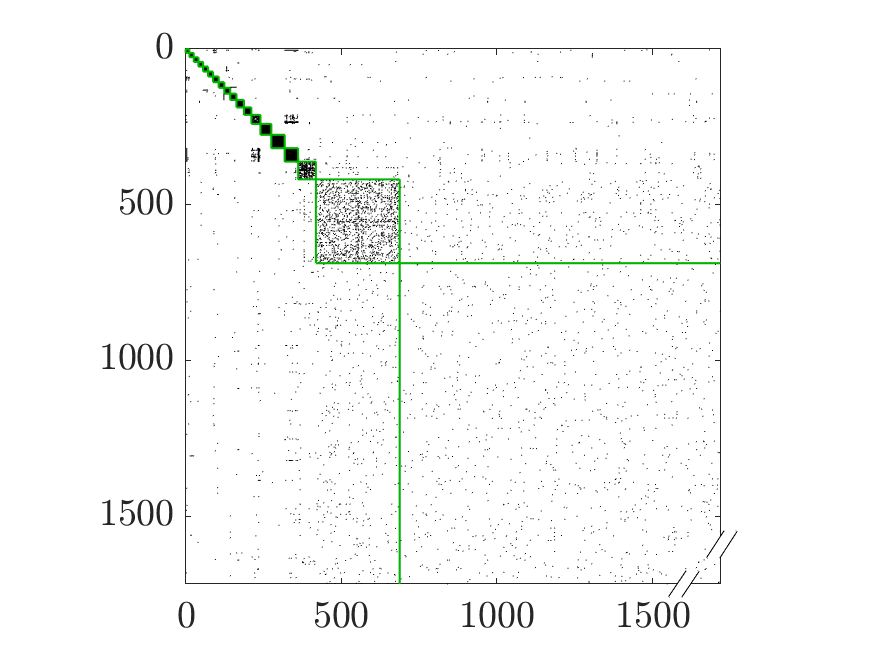}
    \hfill
    \includegraphics[width=0.32\textwidth,trim={1.6cm 0 2.1cm 0},clip]{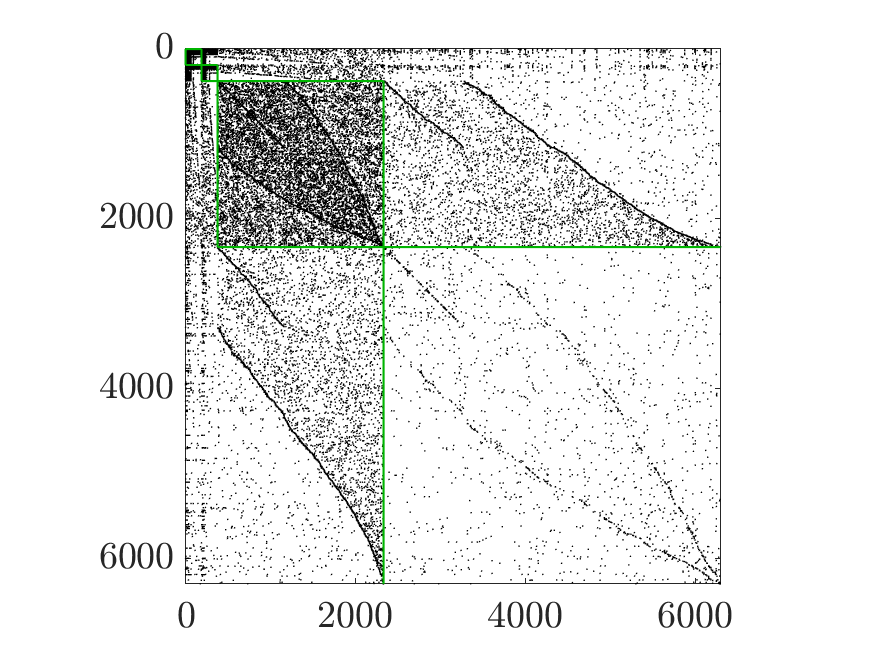}
    \caption{
    Heat map of the adjacency matrix, reordered with respect to the recovered communities (shown with solid lines, for the real-world networks HEP-TH (left), GRQC (middle) and Gnutella (right).
    }
    \label{Fig:adj_mats}
\end{figure}

We hypothesize that the real-world networks presented in this section exhibit noticeable differences between the inverse of the leading nonbacktracking eigenvalue and the estimates that use direct simulations because the networks contain multiple small dense subgraphs. We explore this idea experimentally by implementing research from Krzakala et al.~\cite{krzakala2013spectral}, and applying spectral clustering techniques to the reduced nonbacktracking matrix $H$ in order to detect some community structure. 
We perform $k$-means clustering on the eigenvectors corresponding to the $j$ largest positive real eigenvalues of $H$, where $j$ is the position of the eigenvalue used in our estimate, as recorded in Table~\ref{Tab:Results}.
Figure~\ref{Fig:adj_mats} clearly demonstrates that the HEP-TH network has 3 small dense cores: they are completely connected subgraphs, indicating 
a community of researchers that all collaborated together at some stage.
In the GRQC network, where we choose to recover $j=19$ communities, the figure shows that this network contains 18 small dense subgraphs, many of which are completely connected. Some of the communities have a smaller density of edges, suggesting a reason for the two distinct large drops in the IPR, which correspond to our two possible estimates. The global percolation properties are best represented by excluding all 18 small dense cores, suggesting why $p^*_\text{deloc} = 1/\lambda_{19}$ performs best.
Lastly, the Gnutella network contains 3 communities with a larger density than the remainder of the network. However, as one of these communities consists of nearly a third of the total population, the global percolation properties of the network are unclear. This may explain why the susceptibility function and the second largest component size exhibit multiple peaks with magnitudes of a similar order.

Although it is not always clear how to assign an exact numerical value to the observed percolation transition in real-world networks, we have presented some examples of networks where $p_\text{nbt}$ yields a clear underestimation. Furthermore, our investigation shows that exploring additional sub-leading eigenvalues can uncover more information into the percolation process.
As highlighted in Table~\ref{Tab:Results}, $p_\text{deloc}$ often aligns more closely with $p_\text{slc}$ than with $p_\text{sus}$. This observation is of particular interest to certain applications which value the dynamics of the second largest component more strongly than the susceptibility of the network. Additionally, we used the inverse participation ratio to measure localisation within networks. Alternative metrics, such as the ratio between the $L^2$ and $L^1$ norm, may slightly affect the outcome, becoming more appropriate in certain applications.

\section{Conclusion} \label{Sec:Conclusion}
In this work, we study the percolation transition on large sparse networks. One standard spectral approach $p_\text{nbt}$, which uses the inverse of the leading eigenvalue of the nonbacktracking matrix, typically fails when the mass of the corresponding eigenvector is concentrated on a small subset of nodes. For this reason, we propose an extension of this method that exploits additional sub-leading eigenvalues. We experimentally show that our approach performs better than the classic spectral estimate on a range of synthetic and real-world networks.

Our method of estimating the percolation transition requires identifying the first delocalised vector $\mathbf{u}(H)$, however, we do not provide a precise definition to classify delocalisation: there are many methods to measure vector localisation and different metrics can produce different outcomes. Further work is required to better understand the effect that these metrics have on our estimate.

We also observe experimentally that the nonbacktracking spectrum of real-world networks often contains a larger number of real eigenvalues when compared to synthetic networks. It would be interesting to further investigate whether this observation can uncover additional information regarding the structure of the network.

Finally, we provide a perturbation theorem to analyse our method on a core-periphery network model: we show that using the inverse of the second largest eigenvalue can outperform the standard estimate $p_\text{nbt}$ for a broad range of model parameters. Our theoretical bound for the eigenvalue perturbation, however, appears much larger than additional experimental results suggest. One avenue for future research is to investigate whether a stronger bound can be obtained, as well as applying the bound to a wider range of networks.

\bibliographystyle{comnet}

\appendix
\section{Proof of Theorem~\ref{Thrm: Eigenvalue perturbation}}
\label{App:Theorem_proof}
\begin{proof}[Proof of Theorem~\ref{Thrm: Eigenvalue perturbation}]
    Let $A, \hat{A}, X, \hat{X}$ be as assumed in Theorem~\ref{Thrm: Eigenvalue perturbation}. Let $\mu$ be an eigenvalue of $L$.

    If there exists an eigenvalue $\nu$ of $\hat{L}$ with $\mu = \nu$, then $|\mu - \nu| =0$ and the theorem holds. Now suppose $\mu \neq \nu$.

    Let $R_\mu = \mu^2 I - \mu A - X$ and $S_\mu = \mu^2 I - \mu \hat{A} - \hat{X}$.
    Then
    \begin{align*}
        R_\mu
        = S_\mu + \mu \hat{A} - \mu A + \hat{X} - X
        = S_\mu (I + S_\mu^{-1}(\mu \hat{A} - \mu A + \hat{X} - X)).
    \end{align*}
    \begin{cl}
    \label{claim:1}
    The matrix $R_\mu = \mu^2 I - \mu A - X$ is singular.
    \end{cl}
    We defer the proof of this, and all other claims, to the end of this section.
    \begin{cl}
    \label{claim:2}
    The matrix $S_\mu = \mu^2 I - \mu \hat{A} - \hat{X}$ is non-singular.
    \end{cl}
    As $R_\mu$ is singular and $S_\mu$ is non-singular, then $(I + S_\mu^{-1}(\mu \hat{A} - \mu A + \hat{X} - X))$ must be singular and have a zero eigenvalue. Thus $-1$ is an eigenvalue of $S_\mu^{-1}(\mu \hat{A} - \mu A + \hat{X} - X))$.
    From the spectral radius inequality, the magnitude of each eigenvalue of a matrix is bounded by the spectral norm. Hence
    \begin{align}
    \begin{split}
        1
        \leq \| S_\mu^{-1}(\mu \hat{A} - \mu A + \hat{X} - X) \|_2
        \leq \| S_\mu^{-1} \|_2 \cdot \|(\mu \hat{A} - \mu A + \hat{X} - X))\|_2.
    \label{eq:pert_proof_Smu_1}
    \end{split}
    \end{align}
    
    Now recall that, for each $i$, the matrices $\hat{A}_i$ and $\hat{X}_i$ are commonly diagonalisable by $P_i$ such that $D_{\hat{A}_i} = P_i \hat{A}_i P_i^{-1}$ and $D_{\hat{X}_i} = P_i \hat{X}_i P_i^{-1}$ are diagonal matrices. Furthermore, each diagonal entry of $D_{\hat{A}_i}$ and $D_{\hat{X}_i}$ are the eigenvalues of $\hat{A}_i$ and $\hat{X}_i$, which we denote as $\lambda_j(\hat{A}_i)$ and $\lambda_j(\hat{X}_i)$ respectively. Then
    \begin{align*}
        S_{\mu,i} = \mu^2 I - \mu \hat{A}_i - \hat{X}_i
        = P_i^{-1} (\mu^2 I - \mu D_{\hat{A}_i} - D_{\hat{X}_i}) P_i.
    \end{align*}
    \begin{cl}
    \label{claim:3}
    The eigenvalues of $S_{\mu,i}$ are the complex numbers $\lambda_j(S_{\mu,i}) = \mu^2 - \mu \lambda_j(\hat{A}_i) - \lambda_j (\hat{X}_i).$ 
    \end{cl}
    Using the spectral radius inequality,
    \begin{equation*}
        \|S_{\mu,i} \|_2 \geq \| P_i \|_2 \cdot \| P_i^{-1} \|_2 \cdot \max_j |\mu^2 - \mu \lambda_j(\hat{A}_i) - \lambda_j (\hat{X}_i)|.
    \end{equation*}
    As $S_{\mu,i}$ and $P_i$ are non-singular, we can invert the inequality such that
    \begin{equation*}
        \|S_{\mu,i}^{-1} \|_2 \leq \| P_i^{-1} \|_2 \cdot \| P_i\|_2 \cdot \left( \max_j |\mu^2 - \mu \lambda_j(\hat{A}_i) - \lambda_j (\hat{X}_i)| \right)^{-1}.
    \end{equation*}
    Let $k_i = \argmax_j |\mu^2 - \mu \lambda_j(\hat{A}_i) - \lambda_j (\hat{X}_i)|$. The spectral norm of a block diagonal matrix is bounded by the maximum spectral norm of each individual block. Thus
    \begin{equation}
    \label{eq:pert_proof_Smu_2}
        \|S_{\mu}^{-1} \|_2 \leq \max_i \|S_{\mu,i}^{-1} \|_2 = \max_i \left( \frac{\| P_i^{-1} \|_2 \cdot \| P_i\|_2}{|\mu^2 - \mu \lambda_{k_i}(\hat{A}_i) - \lambda_{k_i} (\hat{X}_i)|} \right).
    \end{equation}
    By combining \eqref{eq:pert_proof_Smu_1} and \eqref{eq:pert_proof_Smu_2}, we get
    \begin{equation*}
        1 \leq \|(\mu \hat{A} - \mu A + \hat{X} - X)\|_2 \cdot \max_i \left( \frac{\| P_i^{-1} \|_2 \cdot \| P_i\|_2}{|\mu^2 - \mu \lambda_{k_i}(\hat{A}_i) - \lambda_{k_i} (\hat{X}_i)|} \right).
    \end{equation*}
    Let $i_*$ be the $i$ where the maximum occurs. Then
    \begin{equation*}
        |\mu^2 - \mu \lambda_{k_{i_*}}(\hat{A}_{i_*}) - \lambda_{k_{i_*}} (\hat{X}_{i_*})| \leq \| P_{i_*}^{-1} \|_2 \cdot \| P_{i_*}\|_2 \cdot \|(\mu \hat{A} - \mu A + \hat{X} - X)\|_2.
    \end{equation*}
    Let $\alpha$ and $\beta$ be the two complex roots of $\mu^2 - \mu \lambda_{k_{i_*}}(\hat{A}_{i_*}) - \lambda_{k_{i_*}} (\hat{X}_{i_*})$, such that
    \begin{equation*}
        (\mu - \alpha)(\mu - \beta) = \mu^2 - \mu \lambda_{k_{i_*}}(\hat{A}_{i_*}) - \lambda_{k_{i_*}} (\hat{X}_{i_*}) = 0.
    \end{equation*}
    \begin{cl}
    \label{claim:4}
    The two complex roots $\alpha$ and $\beta$ are eigenvalues of $\hat{L}$. 
    \end{cl}
    
    Thus
    \begin{equation}
    \label{eq:thrm_proof_inequality}
        (\mu - \alpha)(\mu - \beta) \leq \kappa (P_{i_*}) \cdot \|\mu \hat{A} - \mu A + \hat{X} - X\|_2,
    \end{equation}
    where $\kappa (P_{i_*}) = \| P_{i_*}^{-1} \|_2 \cdot \| P_{i_*}\|_2$ is the condition number of the diagonalisation matrix $P_{i_*}$. This inequality does not hold if both $|\mu - \alpha| > \sqrt{z}$ and $|\mu - \beta| > \sqrt{z}$, where $z$ is the right hand side of \eqref{eq:thrm_proof_inequality}. Hence there exists an eigenvalue $\nu$ of $\hat{L}$ such that
    \begin{equation*}
        |\mu - \nu| \leq \sqrt{\kappa (P_{i_*})} \cdot \sqrt{\|\mu \hat{A} - \mu A + \hat{X} - X\|_2}.
    \end{equation*}
    Observing that $\kappa (P_{i_*}) \leq \max_i \kappa (P_{i})$ completes the proof.
\end{proof}
\\

\begin{proof}[Proof of Claim~\ref{claim:1}]

    Let $\mu$ be an eigenvalue of $L$ and let $\mathbf{v} = \begin{bmatrix}\mathbf{x} &\mathbf{y}\end{bmatrix}^T$ be its eigenvector. Then $L \mathbf{v} = \mu \mathbf{v}$ becomes
    \begin{align*}
        A \mathbf{x} + X \mathbf{y} &= \mu I \mathbf{x},
        \\
        I \mathbf{x} + 0 \mathbf{y} &= \mu I \mathbf{y}.
    \end{align*}
    Substituting the second line into the first gives
    \begin{equation*}
        0 = \mu^2 I \mathbf{y} - \mu A \mathbf{y} - X \mathbf{y} = R_\mu \mathbf{y}.
    \end{equation*}
    Hence $R_\mu$ has a zero eigenvalue and is singular.
\end{proof}
\\

\begin{proof}[Proof of Claim~\ref{claim:2}]

    As $\mu$ is not an eigenvalue of $\hat{L}$, then $\det(\hat{L} - \mu I) \neq 0$. By applying the Schur complement,
    \begin{equation*}
        \det(\hat{L} - \mu I)
        = \det \begin{bmatrix}\hat{A} - \mu I & \hat{X} \\ I  & - \mu I \end{bmatrix} 
        =
        \det (-\mu I) \det \left( \hat{A} - \mu I - \hat{X}(-\mu^{-1}I)I \right) 
        \neq 0.
    \end{equation*}
    Then $\mu \neq 0$ and $\det(\mu^2I - \mu \hat{A} - \hat{X}) = \det(S_\mu) \neq 0$. Hence $S_\mu$ is non-singular.
\end{proof}
\\

\begin{proof}[Proof of Claim~\ref{claim:3}]

    Let $\tilde{\lambda}$ be an eigenvalue of $S_{\mu,i}$, such that $\tilde{\lambda}$ satisfies $\det(S_{\mu,i} - \tilde{\lambda} I) = 0$. Then
    \begin{equation*}
        \det (\mu^2I - \mu \hat{A}_i - \hat{X}_i - \tilde{\lambda} I) 
        =
        \det \left( P_i^{-1} \cdot ( \mu^2I - \mu D_{\hat{A}_i} - D_{\hat{X}_i} - \tilde{\lambda} I ) \cdot P_i \right)
        = 0
    \end{equation*}
    As $P_i$ is non-singular, then
    $$
    \det (\mu^2I - \mu D_{\hat{A}_i} - D_{\hat{X}_i} - \tilde{\lambda} I) 
    =
    \prod_j \left( \mu^2 - \mu \lambda_j (\hat{A}_i) - \lambda_j (\hat{X}_i) - \tilde{\lambda} \right)
    = 0.
    $$
    This equation has the set of solutions
    \begin{equation*}
        \tilde{\lambda} = \mu^2 - \mu \lambda_j (\hat{A}_i) - \lambda_j (\hat{X}_i),
    \end{equation*}
    for all $1 \leq j \leq N$.
\end{proof}
\\

\begin{proof}[Proof of Claim~\ref{claim:4}]
    In the proof of Claim~\ref{claim:2}, we find that non-zero eigenvalues $\mu$ of $\hat{L}$ satisfy $\det( \mu^2I - \mu \hat{A} - \hat{X}) = 0$. As $P_i$ is non-singular, then
    $$
    \prod_i \det(\mu^2 I - \mu D_{\hat{A}_i} - D_{\hat{X}_i}) 
    =
    \prod_{i,j}  \det \left( \mu^2 I - \mu \lambda_j(\hat{A}_i) - \lambda_j(\hat{X}_i) \right)
    = 0.$$
    This has solutions $\mu^2 - \mu \lambda_j(\hat{A}_i) - \lambda_j(\hat{X}_i) = 0$ for all $i$ and $j$. By definition, $\alpha$ and $\beta$ must also be solutions.

\end{proof}

\section{Additional experimental results}
\label{App:Extra_phase}
Figure~\ref{Fig:Syn_Phase_appendix} shows the IPR of the vector $\mathbf{u}_2(H)$, corresponding to the second largest eigenvalue of $H$, where the eigenvalues are instead sorted in descending order of absolute value. The figure corresponds to the experiments shown in Figure~\ref{Fig:Syn_Phase}.
In the regime $d_\text{d} > d_\text{s}$ (above the dashed white line), the model has a core-periphery structure. The vector $\mathbf{u}_2(H)$ is much less localised in the regime $\sqrt{d_\text{d}} < d_\text{s} < d_\text{d}$ (between the white lines), where we expect $1/|\lambda_2(H)|$ to closely predict the percolation transition. In the regime $\sqrt{d_\text{d}} > d_\text{s}$ (above the solid white line), the vector is more localised: the eigenvalue magnitude $|\lambda_2(H)|$ instead reflects the radius of the bulk eigenvalues corresponding to the core subgraph, failing to capture information regarding the global percolation properties.
\begin{figure}[h!]
\vspace{-0.2cm}
    \centering
    \includegraphics[width=0.46\textwidth]
    {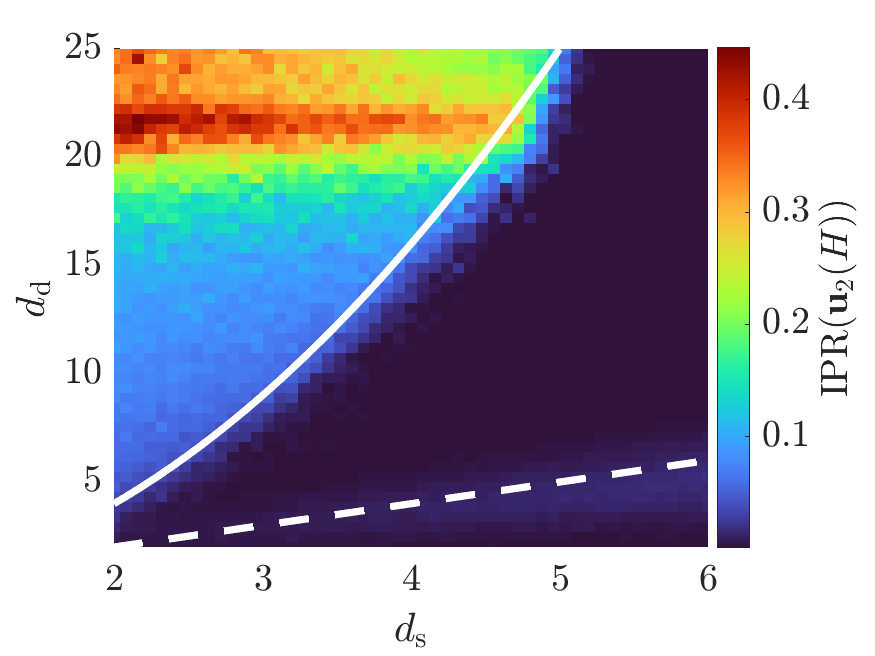}
    \label{Fig:Syn_Phase_appendix}
    \caption{
    Results of an SBM network corresponding to Figure~\ref{Fig:Syn_Phase} using probability matrix $Q$, with $N_\text{d}=25$ and $N_\text{s}=1500$. 
    The inverse participation ratio of the vector $\mathbf{u}_2(H)$ corresponding to the second largest eigenvalue sorted in descending order of absolute value, averaged over 50 realisations.
    Solid line: $d_\text{s} = \sqrt{d_\text{d}}$.
    Dashed line: $d_\text{s} = d_\text{d}$.
    }
    \label{Fig:Syn_Phase_Appendix}
\end{figure}

\end{document}